\documentclass
[
    a4paper,
    DIV=11,
]
{scrartcl}

\usepackage
{
    graphicx,
    amssymb,
    amsmath,
    amsthm,
    dsfont, 
    xcolor,
    mathtools,
    authblk,
    enumitem,
    physics,
    multicol,
    diagbox,
    multirow,
    chemfig
}

\usepackage[utf8]{inputenc}

\usepackage[pdffitwindow=false,
            plainpages=false,
            pdfpagelabels=true,
            pdfpagemode=UseOutlines,
            pdfpagelayout=SinglePage,
            bookmarks=false,
            colorlinks=true,
            hyperfootnotes=false,
            linkcolor=blue,
            urlcolor=blue!30!black,
            citecolor=green!50!black]{hyperref}

\usepackage[bf,normal]{caption}

\graphicspath{{pics/}}

\newcommand{\subfiguretitle}[1]{{\scriptsize{#1}} \\}
\newcommand{\R}{\mathbb{R}}                                     
\newcommand{\pd}[2]{\frac{\partial#1}{\partial#2}}              
\newcommand{\innerprod}[2]{\left\langle #1,\, #2 \right\rangle} 
\newcommand{\ts}{\hspace*{0.1em}}                               
\providecommand{\abs}[1]{\left\lvert #1 \right\rvert}           
\providecommand{\norm}[1]{\left\lVert #1 \right\rVert}          

\newcommand\perb[1]{%
    \mspace{1mu}\vcenter{\hbox{\tikz[x=1ex,y=1ex,line width=0.5pt]
    {\draw (0, -#1-0.5) -- (0,-#1) -- (0,#1) -- (0,#1+0.5);\draw (-0.4,-#1) -- (0.4,-#1);\draw (-0.4,#1) -- (0.4,#1);}}}\mspace{1mu}%
}

\newcommand\xqed[1]{\leavevmode\unskip\penalty9999 \hbox{}\nobreak\hfill \quad\hbox{#1}}
\newcommand{\exampleSymbol}{\xqed{$\triangle$}}

\DeclareMathOperator{\mspan}{span}
\DeclareMathOperator{\sgn}{sgn}
\DeclareMathOperator{\per}{per}
\DeclareMathOperator{\hper}{hper}

\newtheorem{theorem}{Theorem}[section]

\newtheorem{lemma}[theorem]{Lemma}
\newtheorem{proposition}[theorem]{Proposition}
\newtheorem{definition}[theorem]{Definition}
\theoremstyle{definition}
\newtheorem{example}[theorem]{Example}
\newtheorem{remark}[theorem]{Remark}

\makeatletter
\renewcommand*\env@matrix[1][*\c@MaxMatrixCols c]{%
  \hskip -\arraycolsep
  \let\@ifnextchar\new@ifnextchar
  \array{#1}}
\makeatother

\mathtoolsset{centercolon} 

\allowdisplaybreaks

\DeclareCaptionLabelFormat{period}{#1~#2}
\begin{document}

\title{Symmetric and antisymmetric kernels \\ for machine learning problems \\ in quantum physics and chemistry}
\author[1]{Stefan Klus}
\author[2]{Patrick Gelß}
\author[3]{Feliks Nüske}
\author[2,4,5]{Frank No\'e}

\affil[1]{Department of Mathematics, University of Surrey, UK}
\affil[2]{Department of Mathematics and Computer Science, Freie Universität Berlin, Germany}
\affil[3]{Department of Mathematics, Paderborn University, Germany}
\affil[4]{Department of Physics, Freie Universität Berlin, Germany}
\affil[5]{Department of Chemistry, Rice University, Houston TX, USA}

\date{}

\maketitle

\begin{abstract}
We derive symmetric and antisymmetric kernels by symmetrizing and antisymmetrizing conventional kernels and analyze their properties. In particular, we compute the feature space dimensions of the resulting polynomial kernels, prove that the reproducing kernel Hilbert spaces induced by symmetric and antisymmetric Gaussian kernels are dense in the space of symmetric and antisymmetric functions, and propose a Slater determinant representation of the antisymmetric Gaussian kernel, which allows for an efficient evaluation even if the state space is high-dimensional. Furthermore, we show that by exploiting symmetries or antisymmetries the size of the training data set can be significantly reduced. The results are illustrated with guiding examples and simple quantum physics and chemistry applications.
\end{abstract}

\section{Introduction}

Kernel methods and neural networks are two of the most prevalent and versatile machine learning techniques. While various recent publications focus on invariant or equivariant deep learning algorithms, our goal is to derive kernel-based methods that exploit symmetries. Symmetries play an important role in many research areas such as physics and chemistry \cite{Li13, Koner20, Hutter20}, but also point cloud classification problems \cite{Qi17} or problems  defined on sets \cite{Lee19} are naturally permutation-invariant. One of the most prominent applications is in quantum physics. Systems of bosons require symmetric wave functions, whereas systems of fermions are represented by antisymmetric wave functions. Exploiting such symmetries of the underlying system is a popular and powerful approach that has been used to improve the performance of kernel-based methods as well as deep-learning algorithms. The goal is to obtain more accurate representations without increasing the number of training data points---resulting in more efficient learning algorithms---and to ensure that symmetry constraints are satisfied. In \cite{Li13} and \cite{Koner20}, for instance, neural networks and kernel approaches that take into account symmetries of molecules are constructed. These methods are then used for learning potential energy surfaces. An approach for constructing potential energy surfaces based on Gaussian processes combined with permutation-invariant kernels can be found in~\cite{Uteva17}. Gaussian processes that exploit symmetries by summing over permutations of identical atoms are also utilized in \cite{Bartok13a} to improve the accuracy of density functional theory descriptions. Moreover, the so-called SOAP (smooth overlap of atomic positions) kernel~\cite{Bartok13b} is a popular framework to design translation-, rotation-, and permutation-invariant descriptors of molecules. In \cite{Haasdonk07}, general invariant kernels (capturing discrete and continuous transformations) for pattern analysis are defined and analyzed. Recently, neural network architectures for antisymmetric wavefunctions have been proposed \cite{Han19a, Pfau20, Hermann20, Choo20, Schaetzle21} that typically operate by applying Slater determinants to the outputs. The neural networks optimize the basis functions entering the Slater determinants through a deep learning variant of a technique called \emph{backflow}. Backflow is a method to modify the basis functions used in quantum Monte Carlo as trial wavefunctions \cite{LopezRios06}. Neural network approaches such as \emph{FermiNet}~\cite{Pfau20} and \emph{PauliNet}~\cite{Hermann20} achieve extremely high accuracy with relatively few Slater determinants compared to standard quantum chemistry methods that build Slater determinants with fixed basis functions. Kernels, on the other hand, accomplish this by mapping the data to potentially infinite-dimensional feature spaces. Any continuous antisymmetric function can be approximated by antisymmetrized universal kernels. The universal approximation of symmetric and anti-symmetric functions is also studied in \cite{Han19b}.

In this work, we develop kernels that are intrinsically symmetric or antisymmetric. Although we focus mostly on physics and chemistry applications in what follows, the derived kernels can be used in the same way in other kernel-based supervised or unsupervised learning algorithms such as kernel principal component analysis (kernel PCA)~\cite{Schoelkopf98}, kernel canonical correlation analysis (kernel CCA)~\cite{Melzer02}, or support vector machines (SVMs)~\cite{Steinwart08:SVM}. The main contributions are:
\begin{itemize}[leftmargin=1em, itemsep=0ex, topsep=0.5ex]
\item We derive symmetric and antisymmetric kernels based on conventional kernels such as polynomial and Gaussian kernels and show that certain kernels can be expressed as Slater permanents or determinants.
\item We analyze the feature spaces and approximation properties of such kernels.
\item We demonstrate that these techniques improve the efficiency of kernel-based methods for problems exhibiting symmetries or antisymmetries.
\item We apply kernel-based methods for solving the time-independent Schrödinger equation to simple quantum mechanics problems. Furthermore, we predict the boiling points of molecules using kernel ridge regression.
\end{itemize}

In Section \ref{sec:Kernels and kernel-based methods}, we first introduce kernels, reproducing kernel Hilbert spaces, and kernel-based methods for solving the time-independent Schrödinger equation. Antisymmetric kernels will be derived in Section~\ref{sec:Antisymmetric kernels and their properties} and symmetric kernels in Section~\ref{sec:Symmetric kernels and their properties}.
These two sections contain the main theoretical results, in particular the analysis of the properties of the resulting polynomial and Gaussian kernels. Numerical results will be presented in Section~\ref{sec:Applications}. We conclude the paper with a list of open problems and future research.

\section{Kernels and kernel-based methods}
\label{sec:Kernels and kernel-based methods}

We will briefly recapitulate the properties of kernels and introduce the induced reproducing kernel Hilbert spaces. Additionally, we will present a kernel-based method for solving the time-independent Schrödinger equation.

\subsection{Reproducing kernel Hilbert spaces}

A kernel can be regarded as a similarity measure. We will focus on real-valued kernels, but the definitions can be easily extended to complex domains.

\begin{definition}[Kernel \cite{Steinwart08:SVM}]
Given a non-empty set $ \mathbb{X} $, a function $ k \colon \mathbb{X} \times \mathbb{X} $ is called \emph{kernel} if there exists a Hilbert space $ \mathbb{H} $ and a \emph{feature map} $ \phi \colon \mathbb{X} \to \mathbb{H} $ such that
\begin{equation*}
    k(x, x^\prime) = \innerprod{\phi(x)}{\phi(x^\prime)}.
\end{equation*}
\end{definition}

For a given kernel $ k $, the so-called \emph{Gram matrix} $ G \in \R^{m \times m} $ associated with a data set $ \{\ts x^{(i)} \ts\}_{i=1}^m \subset \mathbb{X} $ is defined by $ G_{ij} = k(x^{(i)}, x^{(j)}) $.

\begin{definition}[Positive definiteness \cite{Steinwart08:SVM}]
A function $ k \colon \mathbb{X} \times \mathbb{X} $ is called positive definite if for all $ m $, all vectors $ c = [c_1, \dots, c_m]^\top \in \R^m $, and all subsets $ \{\ts x^{(i)} \ts\}_{i=1}^m \subset \mathbb{X} $ it holds that
\begin{equation*}
    c^\top G \ts c = \sum_{i=1}^{m} \sum_{j=1}^m c_i \ts c_j \ts k(x^{(i)}, x^{(j)}) \ge 0.
\end{equation*}
\end{definition}

Strictly positive definite means that $ c^\top G \ts c = 0 $ for mutually distinct data points only if $ c = 0 $. It can be shown that a function $ k \colon \mathbb{X} \times \mathbb{X} \to \R $ is a kernel if and only if it is symmetric, i.e., $ k(x, x^\prime) = k(x^\prime, x) $, and positive definite (s.p.d.~in what follows to avoid confusion between different notions of symmetry), see~\cite{Steinwart08:SVM}. Such kernels induce so-called reproducing kernel Hilbert spaces.

\begin{definition}[RKHS \cite{Schoe01, Steinwart08:SVM}]
Let $ \mathbb{X} $ be a  non-empty set. A space  $ \mathbb{H} $ of functions $ f \colon \mathbb{X} \to \R $ is called \emph{reproducing kernel Hilbert space} (RKHS) with inner product $ \innerprod{\,\cdot\,}{\,\cdot\,}_\mathbb{H} $ if a kernel $ k $ exists such that
\begin{enumerate}[label=(\roman*), itemsep=0ex, topsep=1ex]
\item $ f(x) = \innerprod{f}{k(x, \,\cdot\,)}_\mathbb{H} $ for all $ f \in \mathbb{H} $, and
\item $ \mathbb{H} = \overline{\mspan\{k(x, \,\cdot\,) \mid x \in \mathbb{X} \}} $.
\end{enumerate}
\end{definition}

The first requirement is called the \emph{reproducing property}. For $ f = k(x, \,\cdot\,) $, this results in $ k(x, x^\prime) = \innerprod{k(x, \,\cdot\,)}{k(x^\prime\,\cdot\,)}_\mathbb{H} $ so that we can define the so-called \emph{canonical feature map} by $ \phi(x) = k(x, \,\cdot\,) $. Additionally, for a data set $ \{\ts x^{(i)} \ts\}_{i=1}^m $, we define $ \Phi = [\phi(x_1), \dots, \phi(x_m)] $ so that $ G = \Phi^\top \Phi $. For more details on kernels and reproducing kernel Hilbert spaces, we refer to \cite{Schoe01, Steinwart08:SVM}. It was shown in \cite{Wendland04, Zhou08} that not only function evaluations but also derivative evaluations can be represented as inner products in the RKHS $ \mathbb{H} $, provided the kernel is sufficiently smooth. Let now $ \alpha = (\alpha_1, \dots, \alpha_d) \in \mathbb{N}_0^d $ be a multi-index. We define $ \abs{\alpha} = \sum_{i=1}^d \alpha_i $ as usual and, for a fixed $ r \in \mathbb{N}_0 $, the index set $ I_r = \{ \alpha \in \mathbb{N}_0^d: \abs{\alpha} \le r \} $. Given a function $ f \colon \mathbb{X} \to \R $, the partial derivative of $ f $ with respect to $ \alpha $ is defined by
\begin{equation*}
    \mathcal{D}^\alpha f = \pd{^{\abs{\alpha}}}{x_1^{\alpha_1} \dots \ts \partial x_d^{\alpha_d}} f.
\end{equation*}

\begin{theorem}[\cite{Wendland04, Zhou08}]
Let $ r \in \mathbb{N}_0 $ be a non-negative number, $ k \in C^{2 \ts r}(\mathbb{X} \times \mathbb{X}) $ a kernel, and $ \mathbb{H} $ the induced RKHS. Then:
\begin{enumerate}[label=(\roman*), itemsep=0ex, topsep=1ex]
\item $ \mathcal{D}^\alpha k(x, \cdot) \in \mathbb{H} $ for any $ x \in \mathbb{X} $ and $ \alpha \in I_r $.
\item $ (\mathcal{D}^\alpha f)(x) = \innerprod{\mathcal{D}^\alpha k(x, \cdot)}{f}_{\mathbb{H}} $ for any $ x \in \mathbb{X} $, $ f \in \mathbb{H} $, and $ \alpha \in I_r $.
\end{enumerate}
In (i) and (ii), the derivative $\mathcal{D}^\alpha$ is understood as acting on the first argument of the kernel~$k$.
\end{theorem}

We will need this property later for the approximation of differential operators. Another question is how \emph{rich} these Hilbert spaces $ \mathbb{H} $ induced by a kernel $ k $ are.

\begin{definition}[Universal kernel \cite{Micchelli06}]
Let $ \mathbb{X} $ be compact and $ C(\mathbb{X}) $ the space of all continuous functions mapping from $ \mathbb{X} $ to $ \R $ equipped with $ \norm{\,\cdot\,}_\infty $. A  kernel $ k $ is called \emph{universal} if the induced RKHS $ \mathbb{H} $ is dense in $ C(\mathbb{X}) $.
\end{definition}

That is, for a function $ f \in C(\mathbb{X}) $, we can find a function $ g \in \mathbb{H} $ such that $ \norm{g-f}_\infty < \varepsilon $ for any $ \varepsilon > 0 $. The Gaussian kernel
\begin{equation*}
    k(x, x^\prime) = \exp\left(-\frac{\norm{x-x^\prime}^2}{2 \ts \sigma^2}\right),
\end{equation*}
for instance, is universal, while the polynomial kernel
\begin{equation*}
    k(x, x^\prime) = (c + x^\top x^\prime)^q
\end{equation*}
is not. We will analyze the properties of these kernels and their symmetrized and antisymmetrized counterparts in more detail below. Various other notions of universality and the relationships between universal and characteristic kernels are discussed in \cite{Sriperumbudur11}. In what follows, we will omit the subscript $ \mathbb{H} $ if it is clear which inner product or norm we are referring to.

\subsection{Kernel-based solution of the Schrödinger equation}
\label{ssec:Kernel-based solution of the Schroedinger equation}

In \cite{KNH20}, we proposed a kernel-based method for the solution of the time-independent Schrö\-dinger equation and the approximation of other differential operators such as the generator of the Koopman operator. We will restrict ourselves to the Schrödinger equation. Let $ V $ be a potential and $ \mathcal{H} = - \tfrac{\hbar^2}{2 \mathbf{m}} \Delta + V $ the Hamiltonian,
where $ \hbar $ is the reduced Planck constant and $ \mathbf{m} $ the mass, then the time-independent Schrödinger equation is defined by
\begin{equation*}
    \mathcal{H} \psi = E \psi.
\end{equation*}
That is, we want to compute eigenfunctions $ \psi $ and the associated eigenvalues $ E $, which correspond to energies of the system. We define
\begin{align*}
    \mathrm{d}\phi(x) &= - \tfrac{\hbar^2}{2 \mathbf{m}} \sum_{l=1}^d \mathcal{D}^{2 \ts e_l} \phi(x) + V(x) \ts \phi(x),
\end{align*}
where $ e_l $ is the $ l $th unit vector, and operators
\begin{equation*}
    \mathcal{C}_{00} = \int \phi(x) \otimes \phi(x) \ts \mathrm{d} \mu(x)
    \quad \text{and} \quad
    \mathcal{C}_{01} = \int \phi(x) \otimes \mathrm{d}\phi(x) \mathrm{d}\mu(x).
\end{equation*}
Here, $ \mathcal{C}_{00} $ is the standard covariance operator (see~\cite{Baker73}) and $ \mathcal{C}_{01} $ contains the action of the Schrödinger operator. Since these integrals typically cannot be computed in practice, we estimate them using $ \mu $-distributed training data $ \{\ts x^{(i)} \ts\}_{i=1}^m $, resulting in the empirical operators
\begin{equation*}
    \widehat{\mathcal{C}}_{00} = \frac{1}{m} \sum_{i=1}^m \phi(x^{(i)}) \otimes \phi(x^{(i)})
    \quad \text{and} \quad
    \widehat{\mathcal{C}}_{01} = \frac{1}{m} \sum_{i=1}^m \phi(x^{(i)}) \otimes \mathrm{d}\phi(x^{(i)}).
\end{equation*}
Assuming that the eigenfunctions can be represented as $ \widehat{\psi} = \Phi \ts u $, i.e., they are contained in the space spanned by the functions $ \{\ts \phi(x^{(i)}) \ts\}_{i=1}^m $, we obtain a matrix eigenvalue problem 
\begin{equation*}
    G_{10} \ts u = \widehat{E} \ts G_{00} \ts u,
\end{equation*}
where the entries of the (generalized) Gram matrices $ G_{00}, G_{10} \in \R^{m \times m} $ are defined by
\begin{equation*}
    \big[G_{00}\big]_{ij} = \left[\phi(x^{(i)})\right](x^{(j)}) = k(x^{(i)}, x^{(j)})
\end{equation*}
and
\begin{equation*}
    \big[G_{10}\big]_{ij}
        = \left[\mathrm{d}\phi(x^{(i)})\right](x^{(j)})  
        = -\tfrac{\hbar^2}{2 \mathbf{m}} \sum_{l=1}^d \mathcal{D}^{2 e_l} k(x_i, x_j) + V(x_i) \ts k(x_i, x_j).
\end{equation*}
Eigenfunctions are then of the form
\begin{equation*}
    \widehat{\psi} = \Phi \ts u = \sum_{i=1}^m u_i \ts k(x^{(i)}, \,\cdot\,).
\end{equation*}
A detailed derivation and numerical results for simple quantum mechanics problems---the quantum harmonic oscillator and the hydrogen atom---can be found in \cite{KNH20}.

\section{Antisymmetric kernels and their properties}
\label{sec:Antisymmetric kernels and their properties}

In this section, we will introduce the notion of \emph{antisymmetric kernels} and define antisymmetric counterparts of well-known kernels such as the polynomial kernel and the Gaussian kernel. Furthermore, we analyze the properties of the resulting reproducing kernel Hilbert spaces. Most results can then be carried over to the symmetric case, which will be studied in Section~\ref{sec:Symmetric kernels and their properties}.

\subsection{Antisymmetric kernels}
\label{sec:Antisymmetric kernels}

Let $ \mathbb{X} \subset \R^d $ be the state space. Furthermore, let $ S_d $ be the symmetric group and $ \pi \in S_d $ a permutation. With a slight abuse of notation, we define $ \pi(x) = [x_{\pi(1)}, \dots, x_{\pi(d)}]^\top $ to be the vector $ x \in \mathbb{X} $ permuted by $ \pi $. A function $ f \colon \mathbb{X} \to \R $ is called antisymmetric if
\begin{equation*}
    f(x) = \sgn(\pi) \ts f(\pi(x)),
\end{equation*}
where $ \sgn(\pi) $ denotes the sign of the permutation $ \pi $, which is $ 1 $ if the number of transpositions is even and $ -1 $ if it is odd. We define the \emph{antisymmetrization operator} $ \mathcal{A} $ by
\begin{equation*}
    (\mathcal{A} f)(x) = \frac{1}{d!} \sum_{\pi \in S_d} \sgn(\pi) f(\pi(x)).
\end{equation*}

\begin{remark}
In the same way, we can consider state spaces of the form $ \mathbb{X} \subset \bigoplus_{i=1}^{d_x} \R^{d_y} $. Functions would then be antisymmetric with respect to permutations of vectors in $ \R^{d_y} $. That is, for $ x = [x_1, \dots, x_{d_x}]^\top $ with $ x_i \in \R^{d_y} $, the permuted vector is then $ \pi(x) = [x_{\pi(1)}, \dots, x_{\pi(d_x)}]^\top $. For typical quantum mechanics applications, for instance, $ d_y = 3 $ (every particle has a position in a three-dimensional space) and $ d_x $ the number of fermions (or bosons in the symmetric case). The special case $ d_x = 2 $ is considered in \cite{Pahikkala15}, where the spectral properties of symmetric and antisymmetric pairwise kernels are analyzed. Supervised learning problems with such pairwise kernels are discussed in \cite{Gnecco18}.
\end{remark}

Our goal is to define antisymmetric kernels for arbitrary $ d $, which can then be used in kernel-based learning algorithms.

\begin{definition}[Antisymmetric kernel function]
Let $ k \colon \mathbb{X} \times \mathbb{X} \to \R $ be a kernel. We define an antisymmetric function $ k_a \colon \mathbb{X} \times \mathbb{X} \to \R $ by
\begin{equation*}
    k_a(x, x^\prime) = \frac{1}{d!^2} \sum_{\pi \in S_d} \sum_{\pi^\prime \in S_d} \sgn(\pi) \ts \sgn(\pi^\prime) \ts k\big(\pi(x), \pi^\prime(x^\prime)\big).
\end{equation*}
\end{definition}

Clearly, if $ k(x, x^\prime) = k(x^\prime, x) $, then also $ k_a(x, x^\prime) = k_a(x^\prime, x) $. Furthermore, for a fixed permutation $ \widehat{\pi} \in S_d $, it holds that
\begin{equation} \label{eq:antisymmetric kernel - fixed permutation}
\begin{split}
    k_a(\widehat{\pi}(x), x^\prime)
        &= \frac{1}{d!^2} \sum_{\pi \in S_d} \sum_{\pi^\prime \in S_d} \sgn(\pi) \ts \sgn(\pi^\prime) \ts k\big(\pi(\widehat{\pi}(x)), \pi^\prime(x^\prime)\big) \\
        &= \sgn(\widehat{\pi}) \frac{1}{d!^2} \sum_{\pi \in S_d} \sum_{\pi^\prime \in S_d} \ts \sgn(\pi \circ \widehat{\pi}) \ts \sgn(\pi^\prime) \ts k\big([\pi \circ \widehat{\pi}](x), \pi^\prime(x^\prime)\big) \\
        &= \sgn(\widehat{\pi}) \frac{1}{d!^2} \sum_{\pi \in S_d} \sum_{\pi^\prime \in S_d} \ts \sgn(\pi) \ts \sgn(\pi^\prime) \ts k\big(\pi(x), \pi^\prime(x^\prime)\big) \\
        &= \sgn(\widehat{\pi}) \ts k_a(x, x^\prime).
\end{split}
\end{equation}
Here, we used the fact that $ \sgn(\widehat{\pi}) = \sgn(\widehat{\pi}^{-1}) $ and $ \sgn(\pi \circ \widehat{\pi}) = \sgn(\pi) \ts \sgn(\widehat{\pi}) $. Additionally, we utilized the property that for a function $ g \colon S_d \to \R $ it holds that $ \sum_{\pi \in S_d} g(\pi) = \sum_{\pi \in S_d} g(\pi \circ \widehat{\pi}) $, which corresponds to a reordering of the summands. Thus, $ k_a $ is antisymmetric in both arguments. From~\eqref{eq:antisymmetric kernel - fixed permutation} it directly follows that $k_a(x, x^\prime) = 0$ if at least two entries of $x$ or $x^\prime$ are equal.\!\footnote{Assume w.l.o.g.\ that $x_i = x_j$ for some indices $i \neq j$. Let $\widehat{\pi}$ be the permutation which only swaps the positions $i$ and $j$, then it holds that $k_a(x, x^\prime) = k_a(\widehat{\pi}(x), x^\prime) = \sgn(\widehat{\pi}) \ts k_a(x, x^\prime) = -k_a(x, x^\prime)$.}

\begin{lemma}
The function $ k_a $ defines an s.p.d.\ kernel.
\end{lemma}

\begin{proof}
We have
\begin{align*}
    k_a(x, x^\prime)
        &= \frac{1}{d!^2} \sum_{\pi \in S_d} \sum_{\pi^\prime \in S_d} \sgn(\pi) \ts \sgn(\pi^\prime) \ts k\big(\pi(x), \pi^\prime(x^\prime)\big) \\
        &= \frac{1}{d!^2} \sum_{\pi \in S_d} \sum_{\pi^\prime \in S_d} \sgn(\pi) \ts \sgn(\pi^\prime) \ts \innerprod{\phi(\pi(x))}{\phi(\pi^\prime(x^\prime))} \\
        &= \innerprod{\frac{1}{d!} \sum_{\pi \in S_d} \sgn(\pi) \ts \phi(\pi(x))}{\frac{1}{d!} \sum_{\pi^\prime \in S_d} \sgn(\pi^\prime) \ts \phi(\pi^\prime(x^\prime))} \\
        &= \innerprod{\phi_a(x)}{\phi_a(x^\prime)},
\end{align*}
where $ \phi_a(x) = \frac{1}{d!} \sum_{\pi \in S_d} \sgn(\pi) \ts \phi(\pi(x)) $. That is, $ k_a $ is a kernel. Symmetry was shown above. To see that the function is positive definite, let $ c = [c_1, \dots, c_m]^\top \in \R^m $ be a coefficient vector and $ \{\ts x^{(i)} \ts\}_{i=1}^m \subset \mathbb{X} $. Then
\begin{align*}
    c^\top G_a \ts c
        &= \sum_{i=1}^m \sum_{j=1}^m c_i \ts c_j \ts k_a(x^{(i)}, x^{(j)})
        = \sum_{i=1}^m \sum_{j=1}^m c_i \ts c_j \innerprod{\phi_a(x^{(i)})}{\phi_a(x^{(j)})} \\
        &= \innerprod{\sum_{i=1}^m c_i \ts \phi_a(x^{(i)})}{\sum_{j=1}^m c_j \ts \phi_a(x^{(j)})}
        = \norm{\sum_{i=1}^m c_i \ts \phi_a(x^{(i)})}^2 \ge 0. \qedhere
\end{align*}

\end{proof}

The antisymmetrized two- and three-dimensional Gaussian kernels are visualized in Figure~\ref{fig:Gaussian kernel}. The feature space mapping of the antisymmetric kernel $ k_a $ is the antisymmetrization operator $ \mathcal{A} $ applied to the feature space mapping of the kernel $ k $.

\begin{figure}
    \centering
    \begin{minipage}[t]{0.45\textwidth}
        \centering
        \subfiguretitle{(a)}
        \vspace*{1ex}
        \includegraphics[width=0.8\textwidth]{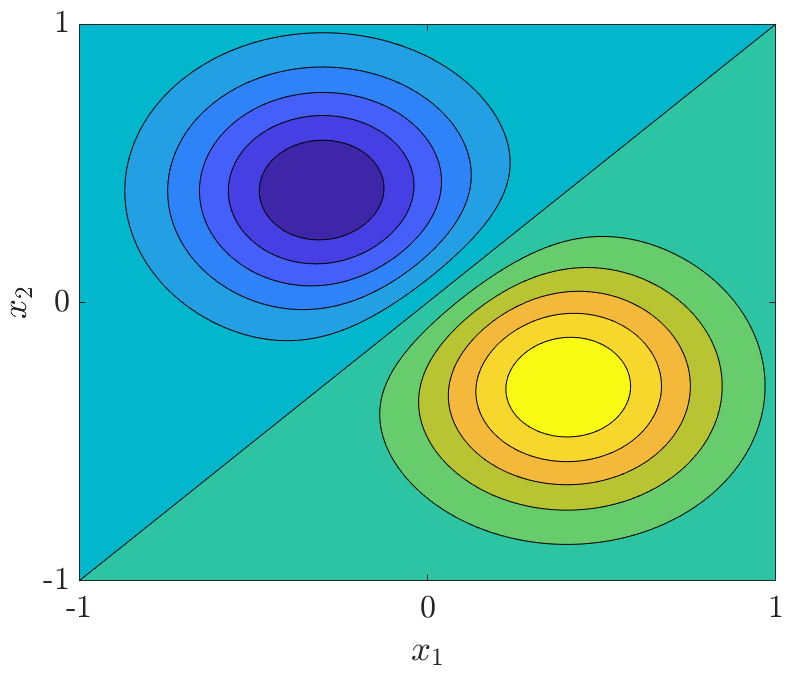}
    \end{minipage}
    \begin{minipage}[t]{0.45\textwidth}
        \centering
        \subfiguretitle{(b)}
        \includegraphics[width=0.9\textwidth]{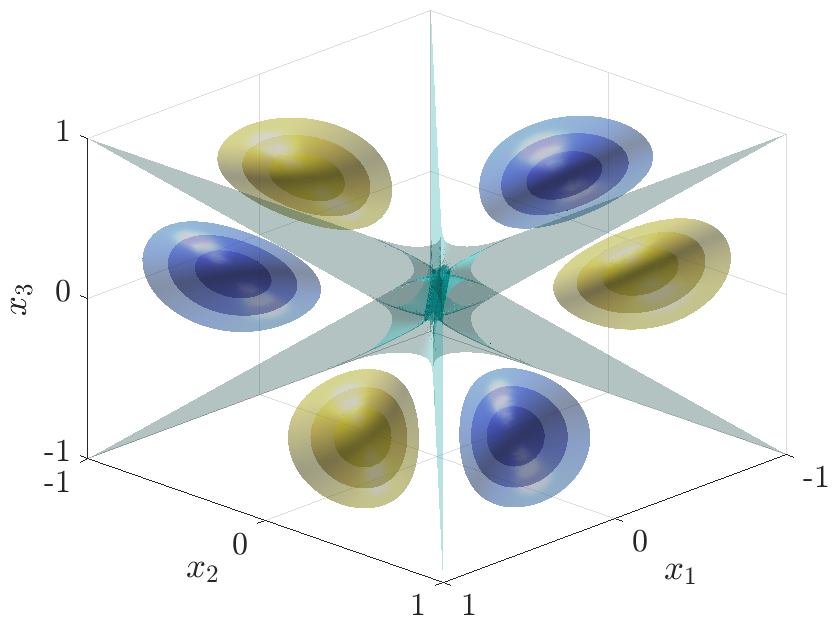}
    \end{minipage}
    \caption{(a) Two-dimensional antisymmetric Gaussian kernel $ k_a $, where $ x^\prime = [0.4, -0.3]^\top $ and $ \sigma = 0.3 $. Yellow corresponds to positive and blue to negative values. (b) Three-dimensional antisymmetric Gaussian kernel $ k_a $, where $ x^\prime = [0.3, -0.6, 0.4]^\top $ and $ \sigma = 0.2 $. The separating isosurface in the middle is defined by $ k_a(x, x^\prime) = 0 $.}
    \label{fig:Gaussian kernel}
\end{figure}

\begin{example} \label{ex:antisymmetrized polynomial kernel}
For $ \mathbb{X} \subset \R^2 $, the feature space of the quadratic kernel $ k(x, x^\prime) = (1 + x^\top x^\prime)^2 $ is spanned by $ \{1, x_1, x_2, x_1^2, x_1 \ts x_2, x_2^2 \} $ and thus six-dimensional. The feature space of the antisymmetrized kernel $ k_a $ is spanned by the two antisymmetric functions $ \{ x_1 - x_2, x_1^2 - x_2^2 \} $. This illustrates that the feature space is significantly reduced. \exampleSymbol
\end{example}

Polynomial kernels of arbitrary degree $ p $ for $ d $-dimensional spaces will be discussed in more detail in Section~\ref{ssec:Antisymmetric polynomial kernels}.

\begin{remark}
The Mercer features of a kernel $ k $ are defined by the eigenfunctions of the integral operator
\begin{equation*}
    (\mathcal{T}_k f)(x) = \int k(x, x^\prime) \ts f(x^\prime) \ts \dd \mu(x^\prime)
\end{equation*}
multiplied by the square root of the associated eigenvalues $ \lambda $, see~\cite{Steinwart08:SVM}. The Mercer features of an antisymmetric kernel $ k_a $ are automatically antisymmetric. This can be seen as follows: Let $ \varphi $ be an eigenfunction of $ \mathcal{T}_{k_a} $ with corresponding eigenvalue $ \lambda $, then
\begin{equation*}
    \lambda \ts \varphi(\pi(x))
        = \int k_a(\pi(x), x^\prime) \ts \varphi(x^\prime) \ts \dd \mu(x^\prime)
        = \sgn(\pi) \int k_a(x, x^\prime) \ts \varphi(x^\prime) \ts \dd \mu(x^\prime) = \lambda \ts \sgn(\pi) \ts \varphi(x).
\end{equation*}
Mercer features of the Gaussian kernel and its antisymmetric and symmetric (see Section~\ref{sec:Symmetric kernels and their properties}) counterparts---computed by a spectral decomposition of the covariance operator, cf.~\cite{Mollenhauer20}---are shown in Figure~\ref{fig:Mercer}.

\begin{figure}
    \centering
    \begin{minipage}{0.25\textwidth}
        \centering
        \subfiguretitle{~~(a)}
        \includegraphics[width=\textwidth]{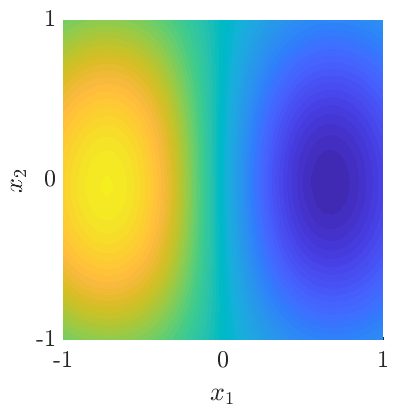} \\
        \includegraphics[width=\textwidth]{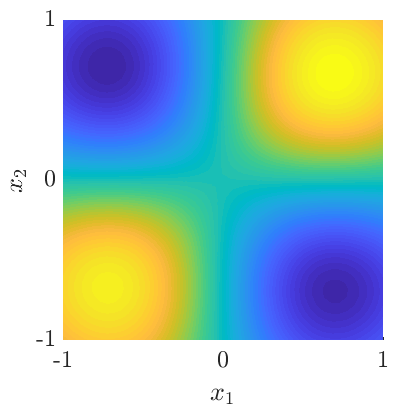}
    \end{minipage}
    \begin{minipage}{0.25\textwidth}
        \centering
        \subfiguretitle{~~(b)}
        \includegraphics[width=\textwidth]{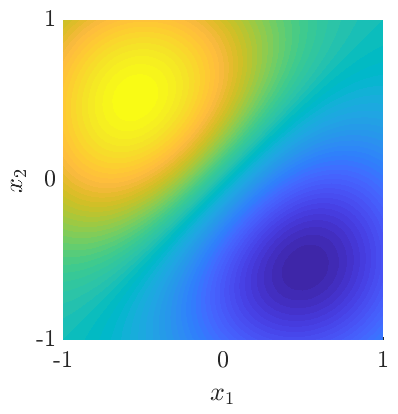} \\
        \includegraphics[width=\textwidth]{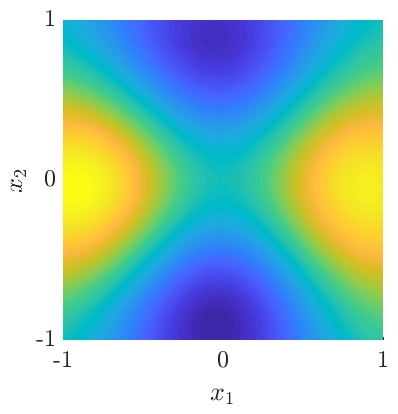}
    \end{minipage}
    \begin{minipage}{0.25\textwidth}
        \centering
        \subfiguretitle{~~(c)}
        \includegraphics[width=\textwidth]{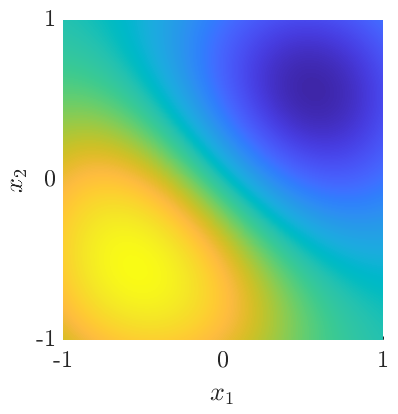} \\
        \includegraphics[width=\textwidth]{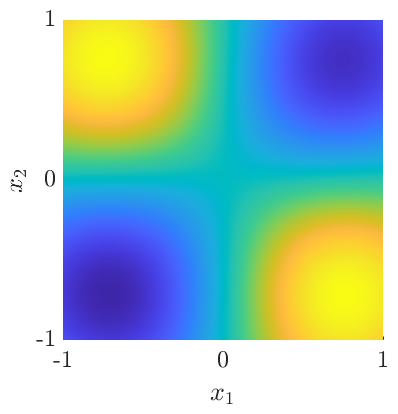}
    \end{minipage}
    \hspace{1ex}
    \begin{minipage}{0.06\textwidth}
        \includegraphics[width=\textwidth]{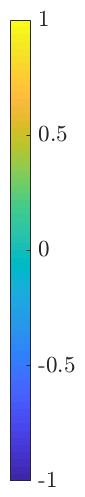}
    \end{minipage}
    \caption{(a) Numerically computed normalized features of the Gaussian kernel $ k $ with bandwidth~$ \sigma = \frac{1}{2} $. (b) Similar-looking but antisymmetric features of the associated kernel~$ k_a $. (c) Symmetric features of the kernel~$ k_s $ derived in Section~\ref{sec:Symmetric kernels and their properties}.}
    \label{fig:Mercer}
\end{figure}

\end{remark}

\begin{definition}[Permutation invariance]
We call a kernel \emph{permutation-invariant} if
\begin{equation*}
    k(x, x^\prime) = k(\pi(x), \pi(x^\prime))
\end{equation*}
for all permutations $ \pi \in S_d $.
\end{definition}

The Gaussian kernel and the polynomial kernel are permutation-invariant since the standard inner product and induced norm are permutation-invariant, i.e., $ \innerprod{x}{x^\prime} = \innerprod{\pi(x)}{\pi(x^\prime)} $ for a permutation $ \pi \in S_d $. The antisymmetric kernel $ k_a $ is permutation-invariant by construction. While many kernels used in practice are naturally permutation-invariant, an open question is whether this assumption limits the expressivity of the induced function space. We will analyze the properties of the Gaussian kernel in Section~\ref{ssec:Antisymmetric Gaussian kernels}. The permutation-invariance allows us to simplify the representation of the antisymmetric kernel.

\begin{lemma} \label{lem:permutation invariant asymmetric kernel}
Given a permutation-invariant kernel $ k $, it holds that
\begin{equation*}
    k_a(x, x^\prime) = \frac{1}{d!} \sum_{\pi \in S_d} \sgn(\pi) \ts k\big(\pi(x), x^\prime\big) = \frac{1}{d!} \ts \sum_{\pi \in S_d} \sgn(\pi) \ts k\big(x, \pi(x^\prime)\big).
\end{equation*}
\end{lemma}

\begin{proof}
We obtain
\begin{align*}
    k_a(x, x^\prime)
        &= \frac{1}{d!^2} \sum_{\pi \in S_d} \sum_{\pi^\prime \in S_d} \sgn(\pi) \ts \sgn(\pi^\prime) \ts k\big(\pi(x), \pi^\prime(x^\prime)\big) \\
        &= \frac{1}{d!^2} \sum_{\pi \in S_d} \sum_{\pi^\prime \in S_d} \sgn(\pi) \ts \sgn(\pi^\prime) \ts k\big(\left[(\pi^\prime)^{-1} \circ \pi\right](x), x^\prime\big) \\
        &= \frac{1}{d!^2} \sum_{\pi \in S_d} \sum_{\pi^\prime \in S_d} \sgn\big((\pi^\prime)^{-1} \circ \pi\big) \ts k\big(\left[(\pi^\prime)^{-1} \circ \pi\right](x), x^\prime\big) \\
        &= \frac{1}{d!} \ts \sum_{\pi \in S_d} \sgn(\pi) \ts k\big(\pi(x), x^\prime\big)
\end{align*}
since all permutations occur $ d! $ times. In the third line, we used the same properties of permutations as above. The proof for the second representation is analogous.
\end{proof}

For the sake of simplicity, assume now that the kernel $ k $ is permutation-invariant. We want to show that for a universal kernel $ k $, the reproducing kernel Hilbert space induced by the corresponding antisymmetric kernel $ k_a $ is dense in the space of antisymmetric functions.

\begin{proposition} \label{pro:universal antisymmetric kernel}
Let $ \mathbb{X} $ be bounded. Given a universal, permutation-invariant, continuous kernel $ k $, the space $ \mathbb{H}_a $ induced by $ k_a $ is dense in the space of continuous antisymmetric functions given by $ C_a(\mathbb{X}) = \{ f \in C(\mathbb{X}) \mid f \text{ is antisymmetric} \} $.
\end{proposition}

\begin{proof}
Let $ f $ be antisymmetric. It follows that $ f(x) = \sgn(\pi) f(\pi(x)) $ for all $ \pi \in S_d $ and thus
\begin{equation*}
    f(x) = \frac{1}{d!} \sum_{\pi \in S_d} \sgn(\pi) f(\pi(x)).
\end{equation*}
Since $ k $ is assumed to be universal, we can find coefficients $ \alpha_i \in \R $ and vectors $ x^{(i)} \in \mathbb{X} $ such that $ \norm{\sum_{i=1}^n \alpha_i \ts k(\,\cdot\,, x^{(i)}) - f}_\infty < \varepsilon $. Then
\begin{align*}
    \norm{\sum_{i=1}^n \alpha_i \ts k_a(\,\cdot\,, x^{(i)}) - f}_\infty
        &= \Bigg\lVert \sum_{i=1}^n \frac{\alpha_i}{d!} \sum_{\pi \in S_d} \sgn(\pi) \ts k\big(\pi(\,\cdot\,), x^{(i)}\big) - \frac{1}{d!} \sum_{\pi \in S_d} \sgn(\pi) f(\pi(\,\cdot\,)) \Bigg\rVert_\infty \\
        &= \Bigg\lVert \frac{1}{d!} \sum_{\pi \in S_d} \sgn(\pi) \underbrace{\left[\sum_{i=1}^n \alpha_i \ts k\big(\pi(\,\cdot\,), x^{(i)}\big) - f(\pi(\,\cdot\,))\right]}_{< \varepsilon ~\forall \pi} \Bigg\rVert_\infty < \varepsilon. \qedhere
\end{align*}
\end{proof}

Continuous antisymmetric functions can be approximated arbitrarily well by universal antisymmetric kernels such as the Gaussian kernel. Although we used the same number of data points for the approximation in the proof (i.e., $ n $ points for the expansion in terms of $ k $ and also $ k_a $), fewer data points are required in practice if we employ the antisymmetric kernel, see Example~\ref{ex:antisymmetric KRR}.

\subsection{Antisymmetric polynomial kernels}
\label{ssec:Antisymmetric polynomial kernels}

We have seen in Example~\ref{ex:antisymmetrized polynomial kernel} that the feature space dimension of the polynomial kernel of order two for $  \mathbb{X} \subset \R^2 $ is reduced from six to two by the antisymmetrization. Let $ q = (q_1, \dots, q_d) \in \mathbb{N}_0^d $ be a multi-index and $ \abs{q} = \sum_{i=1}^{d} q_i $. We define $ x^q = \prod_{i=1}^d x_i^{q_i} $. For a $ d $-dimensional state space $ \mathbb{X} $, the polynomial kernel of order $ p $ is then given by
\begin{equation*}
  k(x, x^\prime)
    = (c + x^\top x^\prime)^p
    = \sum_{0 \le \abs{q} \le p} \left( \sqrt{a_{q}} \ts x^q \right) \left( \sqrt{a_q} \ts {x^\prime}^q \right),
\end{equation*}
where
\begin{equation*}
  a_q = \binom{p}{q} \frac{c^{q_0}}{q_0!}
\end{equation*}
and $ q_0 = p - \abs{q} $, cf.~\cite{Zaki14}. The multinomial coefficients are defined by
\begin{equation*}
    \binom{p}{q} = \frac{p!}{q_1! \dots q_d!}.
\end{equation*}
Thus, the feature space is spanned by the monomials $ \big\{ x^q ~\big|~ 0 \le \abs{q} \le p \big\} $ and the dimension of the feature space is $ n_\phi = \binom{p+d}{d} $, see, e.g., \cite{Shawe-Taylor04}.

We now want to find the feature space of the corresponding antisymmetric kernel $ k_a $. Given a multi-index $ q $, assume that there exist two entries $ q_i $ and $ q_j $ with $ q_i = q_j $. Since the transposition $ (i, j) $ leaves the multi-index (and thus $ x^q $) unchanged, this monomial will be eliminated by the antisymmetrization operator. It follows that the monomials must have distinct indices. In fact, the nonzero images of monomials under antisymmetrization are of the form
\begin{equation} \label{eq:antisymmetric polynomial}
    f_\mu(x) =
    \begin{vmatrix}
        x_1^{\delta_1+\mu_1} & x_1^{\delta_2+\mu_2} & \dots & x_1^{\delta_d+\mu_d} \\
        x_2^{\delta_1+\mu_1} & x_2^{\delta_2+\mu_2} & \dots & x_2^{\delta_d+\mu_d} \\
        \vdots & \vdots & \ddots & \vdots \\
        x_d^{\delta_1+\mu_1} & x_d^{\delta_2+\mu_2} & \dots & x_d^{\delta_d+\mu_d}
    \end{vmatrix},
\end{equation}
where $ \delta = (d-1, d-2, \dots, 0) $ and $ \mu = (\mu_1, \dots, \mu_d) $ with $\mu_1 \geq \mu_2 \geq \dots \geq \mu_d \geq 0$ is a partition of a positive integer,\!\footnote{A \emph{partition} of a positive integer $ n $ is a decomposition into positive integers so that the sum is $ n $. The order of the summands does not matter, i.e., $ 6 = 3 + 2 + 1 $ and $ 6 = 1 + 2 + 3 $ are the same partition. We sort partitions in non-increasing order, e.g., $ \mu = (3, 2, 1) $ is a partition of $ 6 $ into three parts.} see \cite{Sturmfels08}. The degrees of the terms of this antisymmetric polynomial are $ \abs{\mu} + \binom{d}{2} $. Since we need all monomials of order $ 0 \le \abs{q} \le p $, we have to consider the partitions $ \mu $ of $ 0 \le p_r \le p - \binom{d}{2} $.

This representation uses the fact that multi-indices corresponding to antisymmetric polynomials can be written as $ q = \delta + \mu $, where $ \delta $ is defined as above and $ \mu $ a partition. It follows that an antisymmetric polynomial must be at least of order $ \binom{d}{2} $. Equation~\eqref{eq:antisymmetric polynomial} can be regarded as a Slater determinant (introduced below) for a specific set of functions. We also obtain the Vandermonde determinant (up to the sign) as a special case where $ \mu = 0 $.

\begin{definition}[Partition function]
Let $ s_\ell(n) $ be the function that counts the partitions of $ n $ into exactly $ \ell $ parts.
\end{definition}

A closed-form expressions for $ s_\ell(n) $ is not known, but it can be expressed in terms of generating functions or computed using the recurrence relation
\begin{equation*}
    s_\ell(n) = s_\ell(n-\ell) + s_{\ell-1}(n-1),
\end{equation*}
where we define $ s_\ell(n) = 1 $ if $ n = 0 $ and $ \ell = 0 $ and $ s_\ell(n) = 0 $ if $ n \le 0 $ or $ \ell \le 0 $ (but not $ n = \ell = 0 $), see \cite{Stanley11} for more details about partitions and partition functions.

\begin{lemma}
The dimension of the feature space generated by the antisymmetrized polynomial kernel of order $ p $ is
\begin{equation*}
    n_{\phi_a} = \sum_{p_r=0}^{p-\binom{d}{2}} \sum_{j=0}^d s_j(p_r).
\end{equation*}
\end{lemma}

\begin{proof}
Since $ \binom{d}{2} $ of the $ \abs{q} $ exponents are already spoken for, we can use only the remaining $ \abs{q} - \binom{d}{2} $ to generate partitions $ \mu $, with $ 0 \le \abs{q} \le p $. All these numbers can be decomposed into at most $ d $ parts since we have only $ d $ variables. If the number of components is smaller than $ d $, we simply add zeros.
\end{proof}

\begin{example}
For $ d = 3 $ and $ p = 6 $, the base case is $ \delta = (2, 1, 0) $, and we can generate partitions for
\begin{align*}
    p_r = 0: &~ \mu = (0, 0, 0), && q = (2, 1, 0), \\
    p_r = 1: &~ \mu = (1, 0, 0), && q = (3, 1, 0), \\
    p_r = 2: &~ \mu = (2, 0, 0), (1, 1, 0), && q = (4, 1, 0), (3, 2, 0), \\
    p_r = 3: &~ \mu = (3, 0, 0), (2, 1, 0), (1, 1, 1), && q = (5, 1, 0), (4, 2, 0), (3, 2, 1),
\end{align*}
resulting in $ 7 $ antisymmetric polynomials. \exampleSymbol
\end{example}

\begin{table}
\centering
\caption{Dimensions of the feature spaces spanned by the polynomial kernel $ k $ and its antisymmetric counterpart $ k_a $. Here, $ d $ is the dimension of the state space and $ p $ the degree of the polynomial kernel. \\[-0.5ex]}
\label{tab:antisymmetric polynomial kernel}
\scalebox{0.9}{
\begin{tabular}{c|cc|cc|cc|cc|cc|cc|cc}
 \backslashbox{$d$}{$p$} & \multicolumn{2}{c|}{2} & \multicolumn{2}{c|}{3} & \multicolumn{2}{c|}{4} & \multicolumn{2}{c|}{5} & \multicolumn{2}{c|}{6} & \multicolumn{2}{c|}{7} & \multicolumn{2}{c}{8} \\ \hline
& $ n_\phi $ & $ n_{\phi_a} $ & $ n_\phi $ & $ n_{\phi_a} $ & $ n_\phi $ & $ n_{\phi_a} $ & $ n_\phi $ & $ n_{\phi_a} $ & $ n_\phi $ & $ n_{\phi_a} $ & $ n_\phi $ & $ n_{\phi_a} $ & $ n_\phi $ & $ n_{\phi_a} $ \\ \hline
2 & 6 & 2 & 10 & 4 & 15 & 6 & 21 & 9 & 28 & 12 & 36 & 16 & 45 & 20 \\
3 & 10 & 0 & 20 & 1 & 35 & 2 & 56 & 4 & 84 & 7 & 120 & 11 & 165 & 16 \\
4 & 15 & 0 & 35 & 0 & 70 & 0 & 126 & 0 & 210 & 1 & 330 & 2 & 495 & 4
\end{tabular}}
\end{table}

The sizes of the feature spaces of the polynomial kernels $ k $ and $ k_a $ for different dimensions $ d $ and degrees $ p $ are summarized in Table~\ref{tab:antisymmetric polynomial kernel}. This shows that antisymmetric polynomial kernels might not be feasible for higher-dimensional problems. For $ d = 10 $, for example, the lowest degree of the monomials is already $ 45 $.

\subsection{Antisymmetric Gaussian kernels}
\label{ssec:Antisymmetric Gaussian kernels}

We will now analyze the properties of the Gaussian kernel. We have shown in Proposition~\ref{pro:universal antisymmetric kernel} that the space spanned by the antisymmetric Gaussian kernel is dense in the space of continuous antisymmetric functions. For the Gaussian kernel, the expression obtained in Lemma~\ref{lem:permutation invariant asymmetric kernel} can be simplified even further.

\begin{lemma} \label{lem:Gaussian Slater kernel}
Let $ k $ be the Gaussian kernel with bandwidth $ \sigma $, then
\begin{equation*}
    k_a(x, x^\prime) = \frac{1}{d!} \ts
    \begin{vmatrix}
        e^{-\frac{(x_1 - x^\prime_1)^2}{2 \ts \sigma^2}} & \dots & e^{-\frac{(x_1 - x^\prime_d)^2}{2 \ts \sigma^2}} \\
        \vdots & \ddots & \vdots \\
        e^{-\frac{(x_d - x^\prime_1)^2}{2 \ts \sigma^2}} & \dots & e^{-\frac{(x_d - x^\prime_d)^2}{2 \ts \sigma^2}} \\
    \end{vmatrix}.
\end{equation*}
\end{lemma}

\begin{proof}
Applying Leibniz' formula
\begin{equation*}
    \det(A) = \sum_{\pi \in S_d} \sgn(\pi) \prod_{i=1}^d a_{\pi(i), i},
\end{equation*}
we have
\begin{align*}
    \begin{vmatrix}
        e^{-\frac{(x_1 - x^\prime_1)^2}{2 \ts \sigma^2}} & \dots & e^{-\frac{(x_1 - x^\prime_d)^2}{2 \ts \sigma^2}} \\
        \vdots & \ddots & \vdots \\
        e^{-\frac{(x_d - x^\prime_1)^2}{2 \ts \sigma^2}} & \dots & e^{-\frac{(x_d - x^\prime_d)^2}{2 \ts \sigma^2}} \\
    \end{vmatrix}
        &= \sum_{\pi \in S_d} \sgn(\pi) \prod_{i=1}^d e^{-\frac{(x_{\pi(i)} - x^\prime_i)^2}{2 \ts \sigma^2}} \\
        &= \sum_{\pi \in S_d} \sgn(\pi) \ts e^{-\frac{\sum_{i=1}^d (x_{\pi(i)} - x^\prime_i)^2}{2 \ts \sigma^2}} \\
        &= \sum_{\pi \in S_d} \sgn(\pi) \ts e^{-\frac{\norm{\pi(x) - x^\prime}^2}{2 \ts \sigma^2}} \\
        &= \sum_{\pi \in S_d} \sgn(\pi) \ts k(\pi(x), x^\prime).
\end{align*}
Lemma~\ref{lem:permutation invariant asymmetric kernel} then yields the desired result.
\end{proof}

This decomposition is akin to the well-known Slater determinant (see, e.g., \cite{Foldy62}), which defines an antisymmetric wave function by
\begin{equation*}
    \psi_a(x) = \frac{1}{\sqrt{d!}}
    \begin{vmatrix}
        \psi_1(x_1) & \dots & \psi_d(x_1) \\
        \vdots & \ddots & \vdots \\
        \psi_1(x_d) & \dots & \psi_d(x_d) \\
    \end{vmatrix}.
\end{equation*}
Notice that here the normalization factor is chosen in such a way that, provided the wave functions $ \psi_i $, $ i = 1, \dots, d $, are normalized and orthogonal to each other, $ \psi_a $ is normalized as well.

\begin{remark}
We can define a more general class of antisymmetric kernels. Let $ f \colon \R \to \R $ be a function, then
\begin{equation*}
    k_a(x, x^\prime) = \frac{1}{d!} \ts
    \begin{vmatrix}
        f(\abs{x_1 - x^\prime_1}) & \dots & f(\abs{x_1 - x^\prime_d}) \\
        \vdots & \ddots & \vdots \\
        f(\abs{x_d - x^\prime_1}) & \dots & f(\abs{x_d - x^\prime_d}) \\
    \end{vmatrix}
\end{equation*}
defines an antisymmetric kernel. We call such a function $ k_a $ a \emph{Slater kernel}. The Gaussian kernel can be obtained by setting $ f(r) = e^{-\frac{r^2}{2 \ts \sigma^2}} $ and the Laplacian kernel---using the 1-norm---by setting $ f(r) = e^{-\frac{r}{\sigma}} $. Alternatively, kernels based on generalized Slater determinants could be constructed or by concatenating creation and annihilation operators, see also \cite{Pfau20, Hermann20, Choo20}.
\end{remark}

The advantage of the Slater determinant formulation is that we can compute it efficiently using matrix decomposition techniques, without having to iterate over all permutations, which would be clearly infeasible for higher-dimensional problems.

\begin{example} \label{ex:antisymmetric KRR}
In order to illustrate the difference between a standard Gaussian kernel~$ k $ and its antisymmetrized counterpart~$ k_a $, we define an antisymmetric function $ f \colon \R^2 \to \R $ by $ f(x) = \sin(\boldsymbol{\pi}(x_1 - x_2)) $ and apply kernel ridge regression (see, e.g., \cite{Shawe-Taylor04}) to randomly sampled data points.\!\footnote{We use a bold $\boldsymbol{\pi}$ for the mathematical constant to avoid confusion with permutations $ \pi $.} That is, we generate $ m $ data points $ x^{(i)} $ in $ \mathbb{X} = [-1, 1] \times [-1, 1] $ and compute $ y^{(i)} = f(x^{(i)}) $. We then try to recover $ f $ from the training data $ \big\{(x^{(i)}, y^{(i)})\big\}_{i=1}^m $. Additionally, we define an augmented data set of size $ 2 \ts m $ by adding the antisymmetrized data set, i.e., $ \big\{(x^{(i)}, y^{(i)})\big\}_{i=1}^m \cup \big\{(\pi(x^{(i)}), -y^{(i)}\big\}_{i=1}^m $, where $ \pi = (1, 2) $ in cycle notation. The bandwidth of the kernel is set to $ \sigma = \frac{1}{2} $. The results are shown in Figure~\ref{fig:antisymmetric function}. We measure the root-mean-square error (RMSE)---averaged over 5000 runs---in the midpoints of a regular $ 30 \times 30 $ box discretization of the domain. Kernel ridge regression using $ k_a $ results in more accurate function approximations and is, for small $ m $, numerically equivalent to kernel ridge regression using $ k $ applied to the augmented data set of size $ 2 \ts m $. For larger values of $ m $, doubling the size of the data set leads to ill-conditioned matrices and increased numerical errors.\!\footnote{This could be mitigated by decreasing the bandwidth or by regularization techniques.} \exampleSymbol

\begin{figure}
    \centering
    \begin{minipage}{0.45\textwidth}
        \centering
        \subfiguretitle{(a)}
        \includegraphics[width=0.9\textwidth]{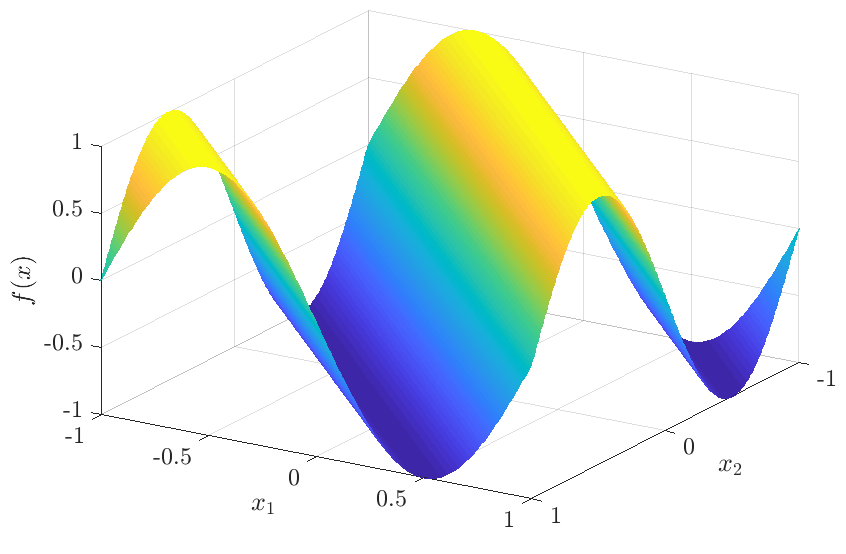}
    \end{minipage}
    \begin{minipage}{0.45\textwidth}
        \centering
        \subfiguretitle{(b)}
        \includegraphics[width=0.9\textwidth]{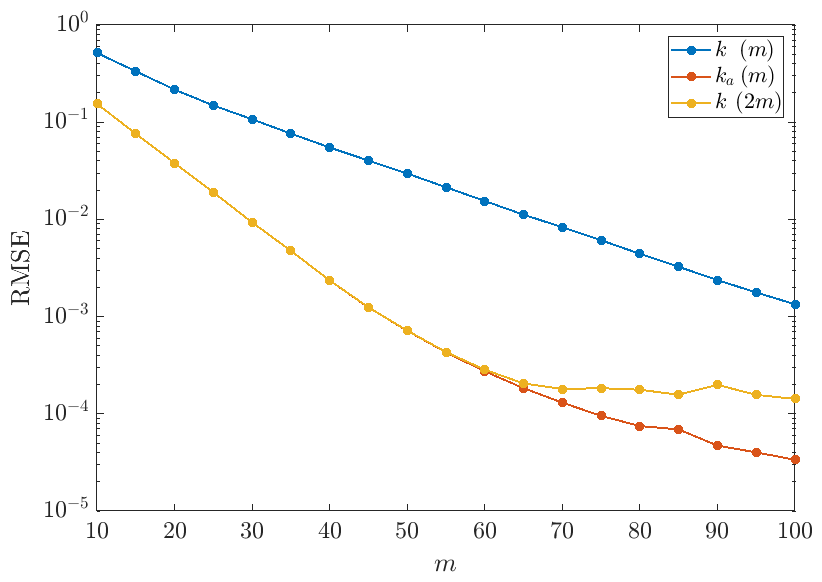}
    \end{minipage}
    \caption{(a) Antisymmetric function $ f(x) = \sin(\boldsymbol{\pi}(x_1 - x_2)) $. (b) Kernel ridge regression approximation error as a function of the number of data points. The antisymmetric Gaussian kernel leads to more accurate function approximations without increasing the size of the training data set.}
    \label{fig:antisymmetric function}
\end{figure}

\end{example}

The example shows that the antisymmetrized kernel is indeed advantageous, it enables a more accurate representation without increasing the size of the data set. For higher-dimensional problems, this effect will be even more pronounced. To obtain the same accuracy for a three-dimensional antisymmetric function, we would already need $ 6 \ts m $ data points. The kernel evaluations, on the other hand, become more expensive, but are easily parallelizable. The bottleneck of kernel-based methods is often the size of the training data set, which enters in a cubic way (since a generally dense system of linear equations has to be solved, or, if we are interested in eigenfunctions of operators associated with dynamical systems, a generalized eigenvalue problem).

\subsection{Derivatives of antisymmetric kernels}

For the approximation of differential operators, we will also need partial derivatives of the kernel $ k_a $. Since $ k_a $ just comprises alternating sums of kernel functions $ k $, we can compute derivatives of $ k_a $ by summing over derivatives of $ k $. For polynomial and Gaussian kernels, the derivatives of $ k $ can be found in~\cite{KNH20}. Alternatively, the partial derivatives of the antisymmetric Gaussian kernel can be computed via Slater determinants.

\begin{example}
For the antisymmetric Gaussian kernel, let $ K^{e_l} \in \R^{d \times d} $ be the matrix with entries
\begin{equation*}
    K^{e_l}_{i j} =
    \begin{cases}
        -\frac{1}{\sigma^2} (x_i - x^\prime_j) e^{-\frac{(x_i - x^\prime_j)^2}{2 \ts \sigma^2}}, & i = l, \\
        e^{-\frac{(x_i - x^\prime_j)^2}{2 \ts \sigma^2}}, & i \ne l.
    \end{cases}
\end{equation*}
Then
\begin{equation*}
    \mathcal{D}^{e_l} k_a(x, x^\prime) = \frac{1}{d!} \ts \det(K^{e_l}).
\end{equation*}
Similar formulas can be derived for the second-order derivatives. \exampleSymbol
\end{example}

\section{Symmetric kernels and their properties}
\label{sec:Symmetric kernels and their properties}

Although we focused on antisymmetric functions so far, symmetric functions also play an important role in quantum physics. Other typical applications include point clouds, sets, and graphs, where the numbering of points, elements, or vertices should not impair the learning algorithms. Some of the above results can be easily carried over to the symmetric case. The special case $ d_x = 2 $ is analyzed in \cite{Pahikkala15}. Similar symmetrized kernels are also constructed in \cite{Uteva17}. We focus on the analysis of the induced functions spaces.

\subsection{Symmetric kernels}

We call a function $ f \colon \mathbb{X} \to \R $ symmetric if
\begin{equation*}
    f(x) = f(\pi(x))
\end{equation*}
for all permutations $ \pi \in S_d $ and define the \emph{symmetrization operator}
\begin{equation*}
    (\mathcal{S} f)(x) = \frac{1}{d!} \sum_{\pi \in S_d} f(\pi(x)).
\end{equation*}

\begin{definition}[Symmetric kernel function]
Let $ k \colon \mathbb{X} \times \mathbb{X} \to \R $ be a kernel. We then define a symmetric function $ k_s \colon \mathbb{X} \times \mathbb{X} \to \R $ by
\begin{equation*}
    k_s(x, x^\prime) = \frac{1}{d!^2} \sum_{\pi \in S_d} \sum_{\pi^\prime \in S_d} k\big(\pi(x), \pi^\prime(x^\prime)\big).
\end{equation*}
\end{definition}

We simply omitted the signs of the permutations here. As before, if $ k(x, x^\prime) = k(x^\prime, x) $, then also $ k_s(x, x^\prime) = k_s(x^\prime, x) $. The function $ k_s $ is permutation-symmetric in both arguments. Note that the definition of permutation-symmetry is different from permutation-invariance, which was defined by $ k(x, x^\prime) = k(\pi(x), \pi(x^\prime)) $. Permutation-symmetric kernels are, however, automatically permutation-invariant. We briefly restate the above results for symmetric functions, the proofs are analogous to their counterparts for antisymmetric functions.

\begin{lemma}
The function $ k_s $ defines an s.p.d.\ kernel.
\end{lemma}

\begin{example} \label{ex:symmetrized polynomial kernel}
For $ \mathbb{X} \subset \R^2 $, the feature space of the symmetrized polynomial kernel of order 2 is spanned by the symmetric functions $ \{ 1, x_1 + x_2, x_1^2 + x_2^2, x_1 \ts x_2 \} $. \exampleSymbol
\end{example}

More general results for polynomial kernels will be derived in Section~\ref{ssec:Symmetric polynomial kernels}. Eigenfunctions of the integral operator associated with $ k_s $ are symmetric. Mercer features of the symmetrized Gaussian kernel for $ d = 2 $ are shown in Figure~\ref{fig:Mercer}.

\begin{lemma} \label{lem:permutation invariant symmetric kernel}
Given a permutation-invariant kernel $ k $, it holds that
\begin{equation*}
    k_s(x, x^\prime) = \frac{1}{d!} \sum_{\pi \in S_d} k\big(\pi(x), x^\prime\big) = \frac{1}{d!} \sum_{\pi \in S_d} k\big(x, \pi(x^\prime)\big).
\end{equation*}
\end{lemma}

Analogously, continuous symmetric functions can be approximated arbitrarily well by symmetric universal kernels.

\begin{proposition}
Let $ \mathbb{X} $ be bounded. Given a universal, permutation-invariant, continuous kernel $ k $, the space $ \mathbb{H}_s $ induced by $ k_s $ is dense in the space of continuous symmetric functions given by $ C_s(\mathbb{X}) = \{ f \in C(\mathbb{X}) \mid f \text{ is symmetric} \} $.
\end{proposition}

\subsection{Symmetric polynomial kernels}
\label{ssec:Symmetric polynomial kernels}

Let us compute the dimensions of the feature spaces spanned by symmetrized polynomial kernels.

\begin{lemma}
The dimension of the feature space generated by the symmetrized polynomial kernel of order $ p $ is
\begin{equation*}
    n_{\phi_s} = \sum_{p_r=0}^{p} \sum_{j=0}^d s_j(p_r).
\end{equation*}
\end{lemma}

\begin{proof}
Let $ \pi $ be a permutation, then the multi-indices $ q $ and $ \pi(q) $ generate the same feature space function when we apply the symmetrization operator $ \mathcal{S} $ to the corresponding monomials $ x^q $ and $ x^{\pi(q)} $. We thus have to consider only partitions $ \mu $ of the integers $ 0 \le \abs{q} \le p $ since the ordering of the multi-indices does not matter.
\end{proof}

This case is similar to the antisymmetric case, with the difference that we require partitions of integers up to $ p $ instead of $ p-\binom{d}{2} $. Table~\ref{tab:symmetric polynomial kernel} lists the dimensions of the feature spaces spanned by the polynomial kernel $ k $ and its symmetric version $ k_s $ for different combinations of $ d $ and $ p $. Compared to the standard polynomial kernel, the number of features is significantly lower, but higher than the number of features generated by the antisymmetric polynomial kernel.

\begin{table}
\centering
\caption{Dimensions of the feature spaces spanned by the polynomial kernel $ k $ and its symmetric counterpart $ k_s $. Here, $ d $ is again the dimension of the state space and $ p $ the degree, cf.~Table~\ref{tab:antisymmetric polynomial kernel}. \\[-0.5ex]}
\label{tab:symmetric polynomial kernel}
\scalebox{0.9}{
\begin{tabular}{c|cc|cc|cc|cc|cc|cc|cc}
 \backslashbox{$d$}{$p$} & \multicolumn{2}{c|}{2} & \multicolumn{2}{c|}{3} & \multicolumn{2}{c|}{4} & \multicolumn{2}{c|}{5} & \multicolumn{2}{c|}{6} & \multicolumn{2}{c|}{7} & \multicolumn{2}{c}{8} \\ \hline
& $ n_\phi $ & $ n_{\phi_s} $ & $ n_\phi $ & $ n_{\phi_s} $ & $ n_\phi $ & $ n_{\phi_s} $ & $ n_\phi $ & $ n_{\phi_s} $ & $ n_\phi $ & $ n_{\phi_s} $ & $ n_\phi $ & $ n_{\phi_s} $ & $ n_\phi $ & $ n_{\phi_s} $ \\ \hline
2 & 6 & 4 & 10 & 6 & 15 & 9 & 21 & 12 & 28 & 16 & 36 & 20 & 45 & 25 \\
3 & 10 & 4 & 20 & 7 & 35 & 11 & 56 & 16 & 84 & 23 & 120 & 31 & 165 & 41 \\
4 & 15 & 4 & 35 & 7 & 70 & 12 & 126 & 18 & 210 & 27 & 330 & 38 & 495 & 53
\end{tabular}}
\end{table}

\subsection{Symmetric Gaussian kernels}
\label{ssec:Symmetric Gaussian kernels}

The symmetric kernel cannot be expressed as a Slater determinant anymore, but we can utilize a related concept. The \emph{permanent} of a matrix $ A \in \R^{d \times d} $ is defined by
\begin{equation*}
    \per(A) = \sum_{\pi \in S_d} \prod_{i=1}^d a_{\pi(i), i} =: \perb{5}
    \begin{matrix}
        a_{11} & \dots & a_{1d} \\
        \vdots & \ddots & \vdots \\
        a_{d1} & \dots & a_{dd}
    \end{matrix}\perb{5}.
\end{equation*}
While for $ d = 2 $ the permanent can be written as a determinant (by flipping the sign of $ a_{12} $ or $ a_{21} $), this is not possible anymore for $ d \ge 3 $ \cite{Szego13}. No polynomial-time algorithm for the computation of the permanent is known, but there are efficient approximation schemes for matrices with non-negative entries \cite{Kuck19}.

\begin{lemma}
Let $ k $ be the Gaussian kernel with bandwidth $ \sigma $, then
\begin{equation*}
    k_s(x, x^\prime) = \frac{1}{d!} \ts
    \perb{7}
    \begin{matrix}
        e^{-\frac{(x_1 - x^\prime_1)^2}{2 \ts \sigma^2}} & \dots & e^{-\frac{(x_1 - x^\prime_d)^2}{2 \ts \sigma^2}} \\
        \vdots & \ddots & \vdots \\
        e^{-\frac{(x_d - x^\prime_1)^2}{2 \ts \sigma^2}} & \dots & e^{-\frac{(x_d - x^\prime_d)^2}{2 \ts \sigma^2}} \\
    \end{matrix}\perb{7}.
\end{equation*}
\end{lemma}

\begin{proof}
The proof is analogous to the one for Lemma~\ref{lem:Gaussian Slater kernel}. Using the definition of the permanent, we obtain
\begin{align*}
    \perb{7}
    \begin{matrix}
        e^{-\frac{(x_1 - x^\prime_1)^2}{2 \ts \sigma^2}} & \dots & e^{-\frac{(x_1 - x^\prime_d)^2}{2 \ts \sigma^2}} \\
        \vdots & \ddots & \vdots \\
        e^{-\frac{(x_d - x^\prime_1)^2}{2 \ts \sigma^2}} & \dots & e^{-\frac{(x_d - x^\prime_d)^2}{2 \ts \sigma^2}} \\
    \end{matrix}
    \perb{7}
        &= \sum_{\pi \in S_d} \prod_{i=1}^d e^{-\frac{(x_{\pi(i)} - x^\prime_i)^2}{2 \ts \sigma^2}} \\
        &= \sum_{\pi \in S_d} e^{-\frac{\sum_{i=1}^d (x_{\pi(i)} - x^\prime_i)^2}{2 \ts \sigma^2}} \\
        &= \sum_{\pi \in S_d} e^{-\frac{\norm{\pi(x) - x^\prime}^2}{2 \ts \sigma^2}} \\
        &= \sum_{\pi \in S_d} k(\pi(x), x^\prime).
\end{align*}
The result then follows from Lemma~\ref{lem:permutation invariant symmetric kernel}.
\end{proof}

\begin{example} \label{ex:graphs}
Assume we have a set of undirected graphs that we would like to classify or categorize. The results should not depend on the vertex labels and thus be identical for isomorphic graphs. Let $ A, A^\prime \in \R^{d \times d} $ be the adjacency matrices of the graphs $ G $ and $ G^\prime $, respectively. We define a Gaussian kernel for graphs by
\begin{equation*}
    k(G, G^\prime) = \exp\left(-\frac{\norm{A-A^\prime}_F^2}{2 \ts \sigma^2}\right),
\end{equation*}
where $ \norm{\,\cdot\,}_F $ denotes the Frobenius norm, and make it symmetric as described above. The only difference here is that we have to define $ \pi(A) = \big(\ts a_{\pi(i),\pi(j)} \ts\big)_{i, j=1}^d $ to permute rows and columns simultaneously. The kernel function $k_s (G, G^\prime)$ can then be expressed in terms of so-called hyperpermanents. We have
\begin{equation*}
    k_s (G, G^\prime) = \sum_{\pi \in S_d} \exp\left(-\frac{\norm{\pi(A)-A^\prime}_F^2}{2 \ts \sigma^2}\right) = \sum_{\pi \in S_d} \prod_{i=1}^d \prod_{j=1}^d \exp\left(-\frac{\big(a_{\pi(i),\pi(j)}- a^\prime_{i,j}\big)^2}{2 \ts \sigma^2}\right).
\end{equation*}
The derivation of a formula for the Laplace expansion of hyperpermanents can be found in Appendix~\ref{app:Laplace}. For the considered example, we set $ \sigma = 1 $ and randomly generate a set of $ 100 $ undirected connected graphs of size $ d = 5 $. We then apply kernel PCA, see \cite{Schoelkopf98}, using the symmetric kernel $ k_s $. Sorting the graphs according to the first principal component, we obtain the ordering shown in Figure~\ref{fig:Graph_kPCA} (only a subset of the graphs is displayed). Isomorphic graphs are grouped into the same category. \exampleSymbol

\begin{figure}
    \centering
    \includegraphics[width=\textwidth]{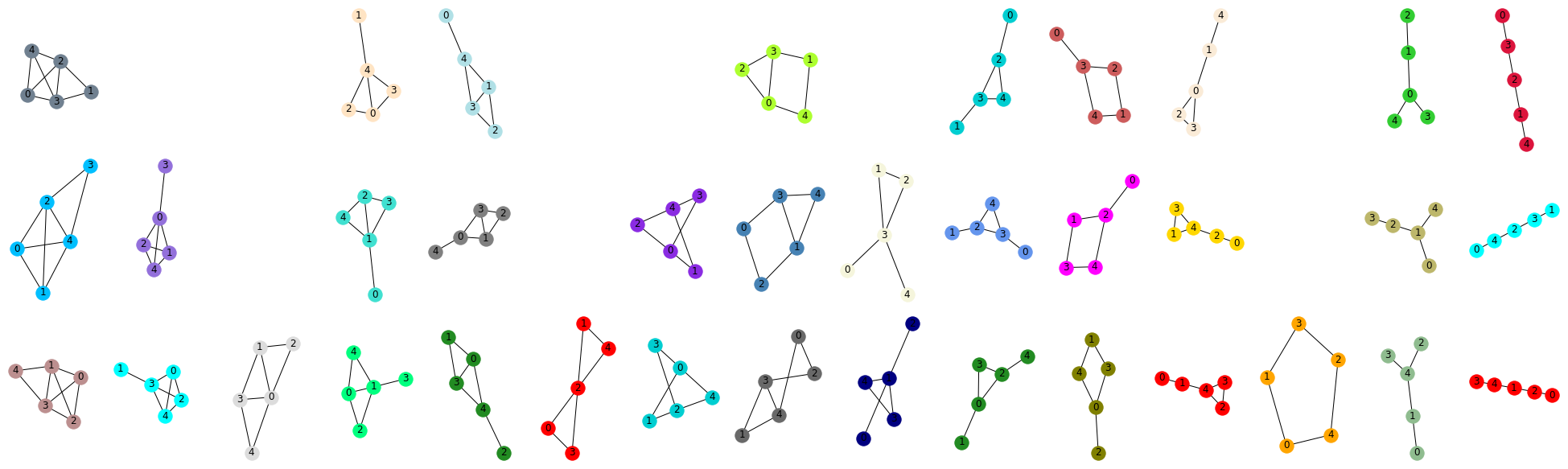}
    \caption{Application of kernel PCA to a set of undirected graphs. The $x$-direction corresponds to the first principal component. The results show that isomorphic graphs are assigned the same value.}
    \label{fig:Graph_kPCA}
\end{figure}

\end{example}

Other learning algorithms such as kernel $k$-means, kernel ridge regression, or support vector machines can be used in the same way, enabling us to cluster, make predictions for, or classify data where the order of elements is irrelevant.

\subsection{Product or quotient representations of symmetric kernels}

The aim now is to express a symmetric kernel not as a Slater permanent but as a product or quotient of antisymmetric functions. As shown in Section~\ref{sec:Antisymmetric kernels}, an antisymmetric kernel is zero for all $ x $ for which a (non-trivial) permutation $ \pi $ exists such that $ \pi(x) = x $. Therefore, products of antisymmetric kernels are zero for such $ x $ as well, see also Figure~\ref{fig:ProductQuotientKernels}. We thus mainly restrict ourselves to quotients. Let $ k_a^{(1)} $ and $ k_a^{(2)} $ be two permutation-invariant antisymmetric kernels and $ k_a^{(2)}(x, x^\prime) \ne 0 $. We define
\begin{equation*}
    k_s(x, x^\prime) = \frac{k_a^{(1)}(x, x^\prime)}{k_a^{(2)}(x, x^\prime)}.
\end{equation*}
Then
\begin{equation*}
    k_s(\pi(x), x^\prime)
        = \frac{k_a^{(1)}(\pi(x), x^\prime)}{k_a^{(2)}(\pi(x), x^\prime)}
        = \frac{\sgn(\pi) \ts k_a^{(1)}(x, x^\prime)}{\sgn(\pi) \ts k_a^{(2)}(x, x^\prime)}
        = k_s(x, x^\prime).
\end{equation*}

\begin{remark}
If the numerator and denominator can be written as determinants, i.e., $ k_a^{(1)}(x, x^\prime) = \det(K_1) $ and $ k_a^{(2)}(x, x^\prime) = \det(K_2) $, we obtain
\begin{equation*}
    k_s(x, x^\prime) = \frac{\det(K_1)}{\det(K_2)} = \det(K_1) \ts \det(K_2^{-1}) = \det(K_1 \ts K_2^{-1}).
\end{equation*}
\end{remark}

\begin{example}\label{ex: quotient Gauss kernel}
Suppose $ d = 2 $. Let $ k^{(1)} $ and $ k^{(2)} $ be two Gaussian kernels with bandwidths $ \sigma_1 $ and $ \sigma_2 $, respectively, where $ \sigma_1 < \sigma_2 $. If the bandwidths are sufficiently small, either $ k^{(1)}(x, x^\prime) $ and $ k^{(2)}(x, x^\prime) $ or $ k^{(1)}(\pi(x), x^\prime) $ and $ k^{(2)}(\pi(x), x^\prime) $ will be close to zero (unless $ \pi(x) = x $), where $ \pi = (1, 2) $ in cycle notation. Assume w.l.o.g.\ the latter holds, then
\begin{equation*}
    k_s(x, x^\prime)
        = \frac{k_a^{(1)}(x, x^\prime)}{k_a^{(2)}(x, x^\prime)}
        = \frac{k^{(1)}(x, x^\prime) - k^{(1)}(\pi(x), x^\prime)}{k^{(2)}(x, x^\prime) - k^{(2)}(\pi(x), x^\prime)}
        \approx \frac{k^{(1)}(x, x^\prime)}{k^{(2)}(x, x^\prime)},
\end{equation*}
which is a Gaussian with bandwidth $ \sigma $ satisfying $ \frac{1}{\sigma^2} = \frac{1}{\sigma_1^2} - \frac{1}{\sigma_2^2} $. This is illustrated in Figure~\ref{fig:ProductQuotientKernels}. Furthermore, the limit of $k_s(x,x^\prime)$ as $x_2 \rightarrow x_1$ exists and is given by
\begin{equation*}
    k_s(x, x^\prime) = \frac{\sigma_2^2}{\sigma_2^1} \frac{k^{(1)}(x, x^\prime)}{k^{(2)}(x, x^\prime)},
\end{equation*}
with $ x = [x_1, x_1]^\top $, see Appendix~\ref{app: quotient kernels}. \exampleSymbol

\begin{figure}
    \centering
    \subfiguretitle{(a) \hspace*{4.1cm} (b) \hspace*{4.1cm} (c)}
    \includegraphics[width=0.9\textwidth]{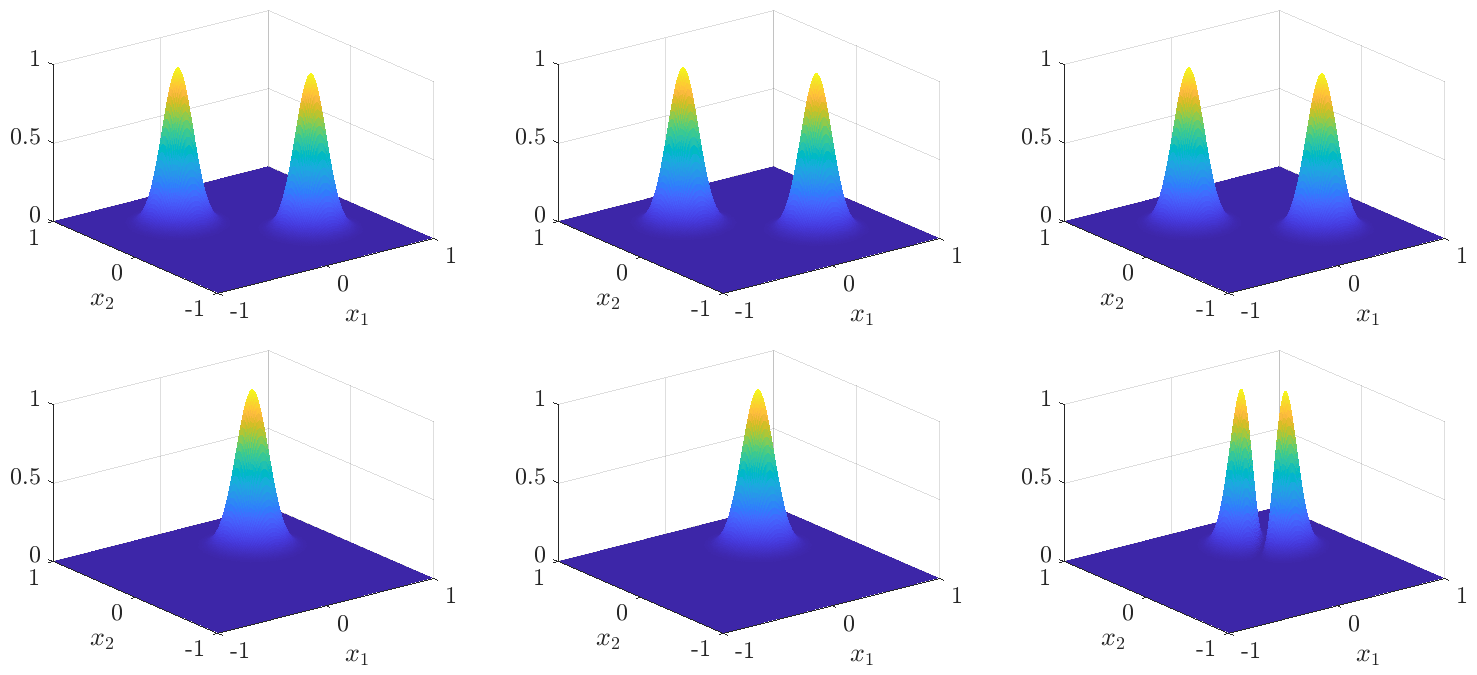}
    \caption{(a)~Symmetric Gaussian kernel. (b)~Quotient of antisymmetric Gaussian kernels. (c)~Product of antisymmetric Gaussian kernels. The bandwidths of the antisymmetric kernels were chosen in such a way that the resulting functions approximate the symmetric Gaussian kernel. In the top row $ x^\prime = [0.4, -0.3]^\top $ and in the bottom row $ x^\prime = [0.4, 0.35]^\top $. Product kernels cannot approximate the symmetric Gaussian kernel if $  x^\prime $ is close to the separating boundary given by $ x_1 = x_2 $.}
    \label{fig:ProductQuotientKernels}
\end{figure}

\end{example}

The symmetric Gaussian kernel can be approximated by a quotient of antisymmetric Gaussian kernels, which can be evaluated in $ \mathcal{O}(d^3) $, thus avoiding the non-polynomial complexity of the permanent. The question whether such kernels are universal is beyond the scope of this work.

\section{Applications}
\label{sec:Applications}

In addition to the guiding examples presented above, we will illustrate the efficacy of the derived kernels with the aid of quantum physics and chemistry problems.

\subsection{Particles in a one-dimensional box}

Let us first consider a simple one-dimensional two-particle system. We define a potential $ V $ by
\begin{equation*}
    V(x) =
    \begin{cases}
        0, & 0 \le x \le L, \\
        \infty, & \text{otherwise}.
    \end{cases}
\end{equation*}
Furthermore, we assume that the two particles do not interact and obtain the Schrödinger equation
\begin{equation} \label{eq:two electrons in a box}
    -\frac{\hbar^2}{2 \ts \mathbf{m}} \Delta \psi(x_1, x_2) = E \ts \psi(x_1, x_2)
\end{equation}
for $ 0 \le x_1, x_2 \le L $. By separating the two variables, we obtain the classical particle in a box problem, with eigenvalues $ E_\ell = \frac{\hbar^2 \pi^2 \ell^2}{2 \mathbf{m} L^2} $ and eigenfunctions $ \psi_\ell(x) = \sqrt{\frac{2}{L}} \sin\left( \frac{\pi \ell \ts x}{L} \right) $, for $ \ell = 1, 2, 3, \dots $, see, for instance, \cite{Hall13}. For the two-particle system, the eigenvalues are hence of the form
\begin{equation*}
    E_{\ell_1, \ell_2} = E_{\ell_1} + E_{\ell_2} = \frac{\hbar^2 \pi^2 (\ell_1^2 + \ell_2^2)}{2 \mathbf{m} L^2}
\end{equation*}
and the eigenfunctions are
\begin{equation*}
    \psi_{\ell_1, \ell_2}(x_1, x_2) = \psi_{\ell_1}(x_1) \ts \psi_{\ell_2}(x_2) = \frac{2}{L} \sin\left( \frac{\pi \ell_1 \ts x_1}{L} \right) \sin\left( \frac{\pi \ell_2 \ts x_2}{L} \right).
\end{equation*}
However, since the two particles are physically indistinguishable, the wave functions must satisfy $ \abs{\psi_{\ell_1, \ell_2}(x_1, x_2)}^2 = \abs{\psi_{\ell_1, \ell_2}(x_2, x_1)}^2 $, which implies that the functions are either symmetric (if the particles are bosons) or antisymmetric (if the particles are fermions). Let us assume that the two particles are electrons, i.e., fermions. We thus want to compute antisymmetric solutions of the time-independent Schrödinger equation by applying the approach introduced in Section~\ref{ssec:Kernel-based solution of the Schroedinger equation}, see also \cite{KNH20}. In the same way, we could assume that the particles are bosons and compute symmetric solutions by replacing the antisymmetric kernel by a symmetric kernel.

We define $ \hbar = 1 $, $ \mathbf{m} = 1 $,  and $ L = \boldsymbol{\pi} $, choose the antisymmetric Gaussian kernel with bandwidth $ \sigma = 0.1 $, and generate $ m = 900 $ uniformly sampled points in $ [0, L] \times [0, L] $. Additionally, to ensure that the eigenfunctions are zero outside the box, we place $ 124 $ equidistantly distributed test points on the boundary and enforce $ \psi_{\ell_1, \ell_2}(x_1, x_2) = 0 $ for these boundary points. We thus have to solve a constrained eigenvalue problem and use the algorithm described in \cite{Golub73}. The first three eigenfunctions $ \psi_{1, 2} $, $ \psi_{1, 3} $, and $ \psi_{2, 3} $ are shown in Figure~\ref{fig:TwoElectrons} and good approximations of the analytically computed eigenfunctions. The probability that the two electrons are in the same location is always zero. Furthermore, the results show that by increasing the number of data points we obtain more accurate and less noisy estimates of the true eigenvalues.

\begin{remark}
We would like to point out that
\begin{itemize}[leftmargin=1em, itemsep=-0.5ex, topsep=0.5ex]
\item this example is just meant as an illustration of the concepts and not as a realistic physical model;
\item eigenfunctions with $ \ell_1 = \ell_2 $ are symmetric and eliminated by the antisymmetrization operation;
\item approximations of the eigenfunctions can be obtained using far fewer points ($ m < 50 $), but the eigenvalues will be considerably overestimated (kernels tailored to quantum mechanics applications might lead to better approximations);
\item the antisymmetry assumption is encoded only in the kernel, not in the Schrödinger equation itself.
\end{itemize}
\end{remark}

\begin{figure}
    \centering
    \begin{minipage}{0.22\textwidth}
        \centering
        \subfiguretitle{~~(a)}
        \includegraphics[width=\textwidth]{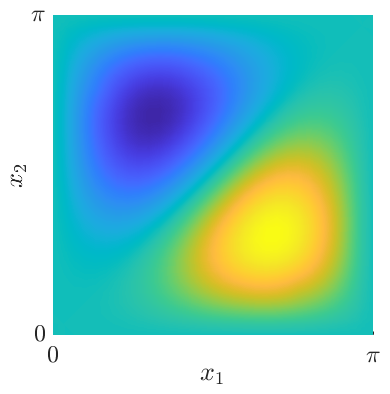}
    \end{minipage}
    \begin{minipage}{0.22\textwidth}
        \centering
        \subfiguretitle{~~(b)}
        \includegraphics[width=\textwidth]{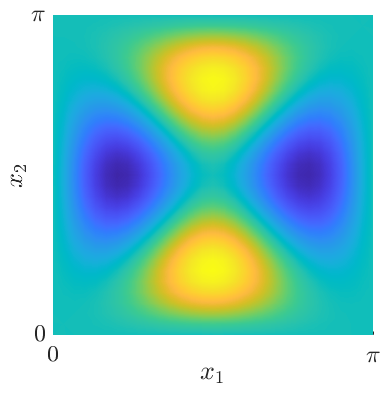}
    \end{minipage}
    \begin{minipage}{0.22\textwidth}
        \centering
        \subfiguretitle{~~(c)}
        \includegraphics[width=\textwidth]{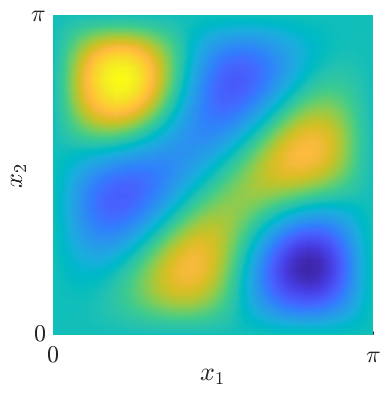}
    \end{minipage}
    \begin{minipage}{0.29\textwidth}
        \centering
        \subfiguretitle{~~(d)}
        \includegraphics[width=\textwidth]{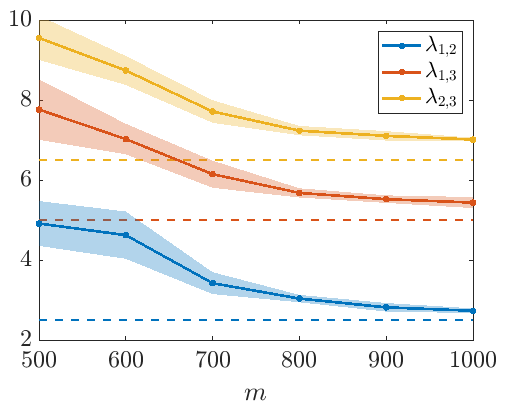}
    \end{minipage}
    \caption{Numerically computed antisymmetric eigenfunctions (a) $ \psi_{1, 2} $, (b) $ \psi_{1, 3} $, and (c)~$ \psi_{2, 3} $ with corresponding eigenvalues $ \lambda_{1, 2} \approx 2.75 $, $ \lambda_{1, 3} \approx 5.40 $, and $ \lambda_{2, 3} \approx 7.07 $ for $ m = 900 $. The eigenvalues are slightly larger than the analytically computed values $ \lambda_{1, 2} = 2.5 $, $ \lambda_{1, 3} = 5 $, and $ \lambda_{2, 3} = 6.5 $. (d)~Eigenvalues as a function of the number of data points. The solid lines represent the numerically computed eigenvalues, the shaded areas the standard deviation, and the dashed lines the analytically computed eigenvalues.}
    \label{fig:TwoElectrons}
\end{figure}

This can be easily extended to the multi-particle case. We now add electron--electron interaction terms resulting in the Hamiltonian
\begin{equation*}
    \mathcal{H} = -\frac{\hbar^2}{2 \ts \mathbf{m}} \Delta + \sum_{i \ne j} \frac{1}{\abs{x_i - x_j}}.
\end{equation*}
For $ d = 3 $, we randomly generate 3000 interior points and 600 boundary points to enforce Dirichlet boundary conditions. We choose a Gaussian kernel with bandwidth $ \sigma = 0.1 $, assemble the Gram matrices, and again solve the resulting constrained eigenvalue problem. The results are shown in Figure~\ref{fig:ThreeElectrons}. For the sake of comparison, we also plot the corresponding eigenfunctions of the Schrödinger equation without the electron--electron interaction. It can be seen that the eigenfunctions for the separable case are similar to the eigenfunctions where the interaction terms are included. For this particular system, the interaction terms do not seem to have a drastic effect on the system's low-lying energy states. In general, however, their effect on the electronic wavefunction can be significant. We also remark that energies and wavefunctions of the interacting system could in principle be approximated by perturbation techniques. However, due to the degeneracy of the antisymmetric states, such a perturbation analysis seems beyond the scope of this work.

\begin{figure}
    \centering
    \subfiguretitle{(a) \hspace*{4.1cm} (b) \hspace*{4.1cm} (c)}
    \includegraphics[width=0.9\textwidth]{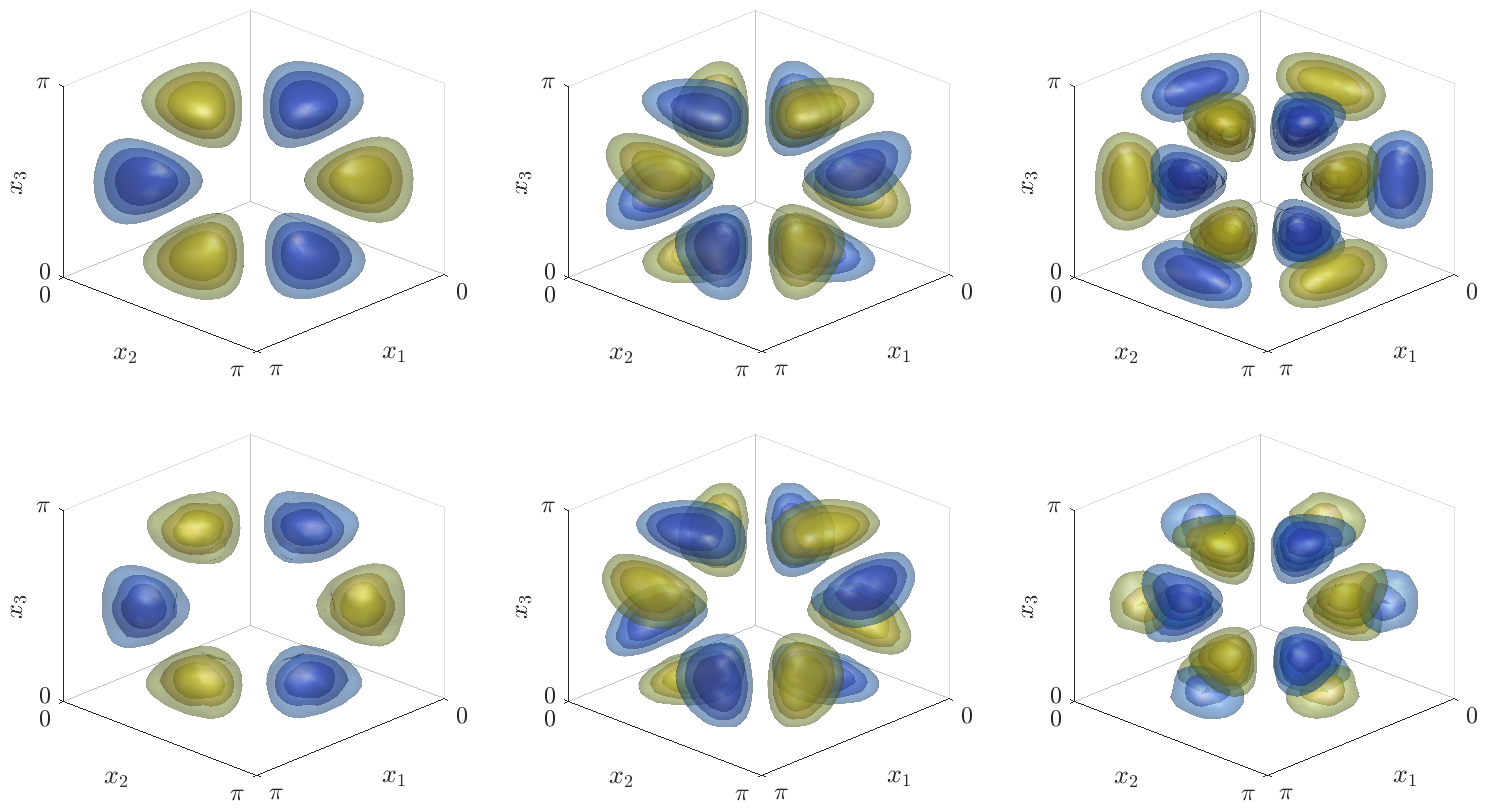}
    \caption{Antisymmetric eigenfunctions (a) $ \psi_{1, 2, 3} $, (b) $ \psi_{1, 2, 4} $, and (c) $ \psi_{1, 3, 4} $. The top row shows the analytically computed eigenfunctions omitting electron--electron interaction, the bottom row the numerically computed eigenfunctions including repulsive forces.}
    \label{fig:ThreeElectrons}
\end{figure}

\subsection{Acyclic molecules}
\label{sec:acyclic}

As a second example, we consider a data set of acyclic molecules~\cite{GAUZERE2012}. The aim is to determine the boiling points of these molecules containing the elements C, H, O, and S. The data set\footnote{The data set can be found at \url{https://brunl01.users.greyc.fr/CHEMISTRY/}.} consists of $183$ graphs $G=(V,E)$ representing the molecular structures and the corresponding boiling points in degrees Celsius, see Figure~\ref{fig:acyclic} for a few examples of molecules included in the data set. The number of vertices $\abs{V}$ varies between $3$ and $11$, where the hydrogen atoms of the molecules are neglected. Thus, in order to compare the graphs of different sizes, we expand all adjacency matrices to $\R^{d \times d}$ with $d=11$ by appending rows and columns of zeros, representing artificial isolated nodes.

We define a symmetrized Laplacian kernel on graphs, cf.~Example~\ref{ex:graphs}. Given the adjacency matrices $A, A^\prime \in \R^{d \times d}$ of the graphs $G=(V,E)$ and $ G^\prime=(V^\prime, E^\prime)$ as well as the kernel parameter $\sigma>0$, we define the tensor $T \in \R^{d \times d \times d \times d}$ by
\begin{equation*}
  t_{i,j,k,l} = \exp \left( - \frac{\abs{ a_{i,j} - a^\prime_{k,l} }}{\sigma}\right)
\end{equation*}
for $i \neq j$ and $k \neq l$ and
\begin{equation*}
  t_{i,i,k,k} = \begin{cases} 1, & \text{if $V_i$ and $V^\prime_k$ represent the same atom,} \\ \exp \left( - 1/\sigma \right), & \text{otherwise.} \end{cases}
\end{equation*}
The latter definition ensures that we avoid unwanted effects of any ordinal labeling of the nodes.
Using the hyperpermanent of $T$, the kernel evaluation $k_s(G,G^\prime)$ can be written as
\begin{equation}\label{eq:graph kernel by hyperpermanent}
  k_s(G,G^\prime) = \hper(T) = \sum_{\pi \in S_d} \prod_{i=1}^d \prod_{j=1}^d t_{i, j, \pi(i), \pi(j)}.
\end{equation}
Note that we do not consider entries $t_{i,j,k,l}$ with either $i=j, k \neq l$ or $i \neq j, k=l$ since $i=j \Leftrightarrow \pi(i)=\pi(j)$. We refer to Appendix~\ref{app:hyperpermanents} for different methods and simplifications for computing the hyperpermanent of~$T$.

\begin{figure}
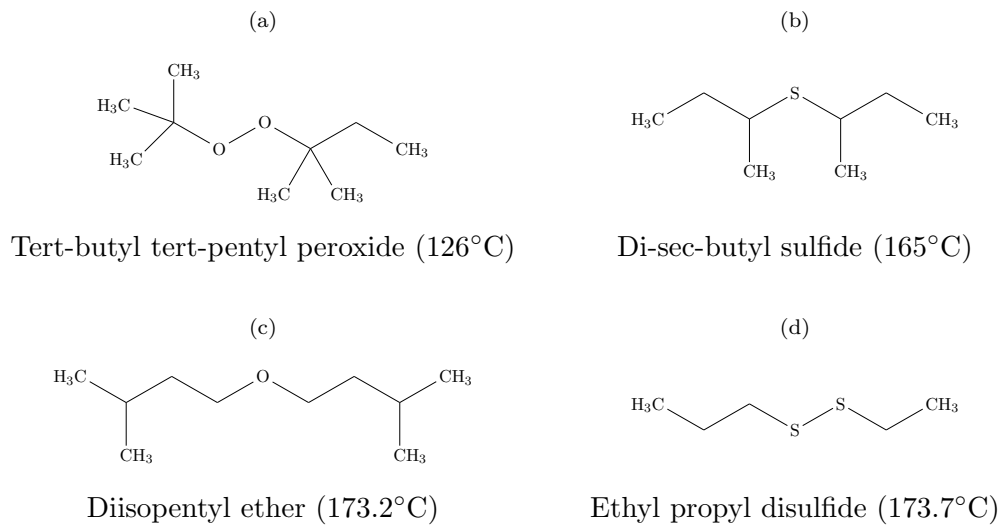

    \centering
    \begin{minipage}{0.5\textwidth}
        \centering
        \scriptsize{(a)}
    \end{minipage}
    \begin{minipage}{0.4\textwidth}
        \centering
        \scriptsize{(b)}
    \end{minipage} \\[1em]
    \begin{minipage}{0.5\textwidth}
        \centering
        \scalebox{0.6}{\chemfig{H_3C-[:45](-[:160]H_3C)(-[:90]CH_3)-[:-30]O-[:30]O-[:-30](-[:-60]CH_3)(-[:-120]H_3C)-[:30]-[:-30]CH_3}}
    \end{minipage}
    \begin{minipage}{0.4\textwidth}
        \centering
        \scalebox{0.6}{\chemfig{H_3C-[:30]-[:-30](-[:-90]CH_3)-[:30]S-[:-30](-[:-90]CH_3)-[:30]-[:-30]CH_3}}
    \end{minipage} \\[1em]
    \begin{minipage}{0.5\textwidth}
        \centering
        Tert-butyl tert-pentyl peroxide (126$^\circ$C) 
    \end{minipage}
    \begin{minipage}{0.4\textwidth}
        \centering
        Di-sec-butyl sulfide (165$^\circ$C) 
    \end{minipage} \\[1.5em]
    \begin{minipage}{0.5\textwidth}
        \centering
        \scriptsize{(c)}
    \end{minipage}
    \begin{minipage}{0.4\textwidth}
        \centering
        \scriptsize{(d)}
    \end{minipage} \\[1em]
    \begin{minipage}{0.5\textwidth}
        \centering
        \scalebox{0.6}{\chemfig{H_3C-[:-30](-[:-90]CH_3)-[:30]-[:-30]-[:30]O-[:-30]-[:30]-[:-30](-[:-90]CH_3)-[:30]CH_3}}
    \end{minipage}
    \begin{minipage}{0.4\textwidth}
        \centering
        \scalebox{0.6}{\chemfig{H_3C-[:-30]-[:30]-[:-30]S-[:30]S-[:-30]-[:30]CH_3}}
    \end{minipage} \\[1em]
    \begin{minipage}{0.5\textwidth}
        \centering
        Diisopentyl ether (173.2$^\circ$C) 
    \end{minipage}
    \begin{minipage}{0.4\textwidth}
        \centering
        Ethyl propyl disulfide (173.7$^\circ$C) 
    \end{minipage}
    \caption{Skeletal formulas of a selection of samples taken from the data set. The set contains oxygen and sulfur compounds of different complexities. The associated boiling points are between $-23.7^\circ$C and $250^\circ$C.}
    \label{fig:acyclic}
\end{figure}

For kernel-based (ridge) regression (see, e.g., \cite{MURPHY2012}), we extract $165$ adjacency matrices ($\approx 90\%$) and their corresponding boiling point temperatures from the data set as training samples, the other data pairs constitute the test set. That is, for any $G$ in the test set, the regression function is given by $f(G) = \Theta^\top K_{\text{train},G}$, where the vector $K_{\text{train}, G} \in \R^{165}$ is the Gram matrix (or \emph{kernel matrix}) corresponding to the training samples (rows) and the test sample $G$ (column). The vector $\Theta \in \R^{165}$ is the solution of $K_{\text{train},\text{train}} \Theta = b$ with $b$ being the vector of boiling points of the molecules in the training set. We then compute the average error as well the root-mean-square error in the boiling points of the test set in order to evaluate the generalizabilty of the learned regression function. We repeat each experiment $10000$ times with randomly chosen training and test sets, the results for different kernel parameters $\sigma$ are shown in Figure~\ref{fig:acyclic_results}~(a).

\begin{figure}
    \centering
    \begin{minipage}{0.45\textwidth}
        \centering
        \subfiguretitle{(a)}
        \includegraphics[height=5cm]{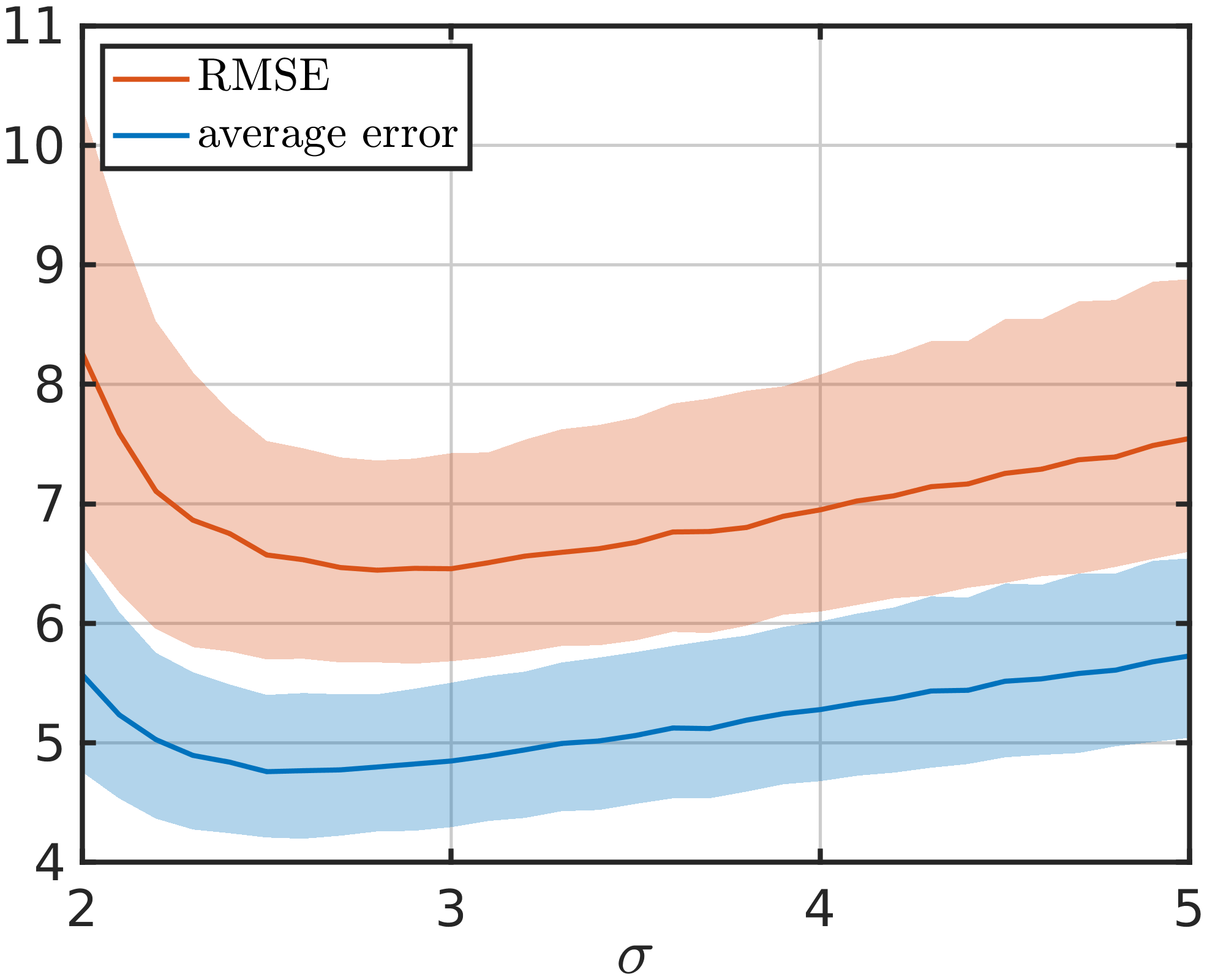}
    \end{minipage}
    \begin{minipage}{0.45\textwidth}
        \centering
        \subfiguretitle{(b)}
        \includegraphics[height=5cm]{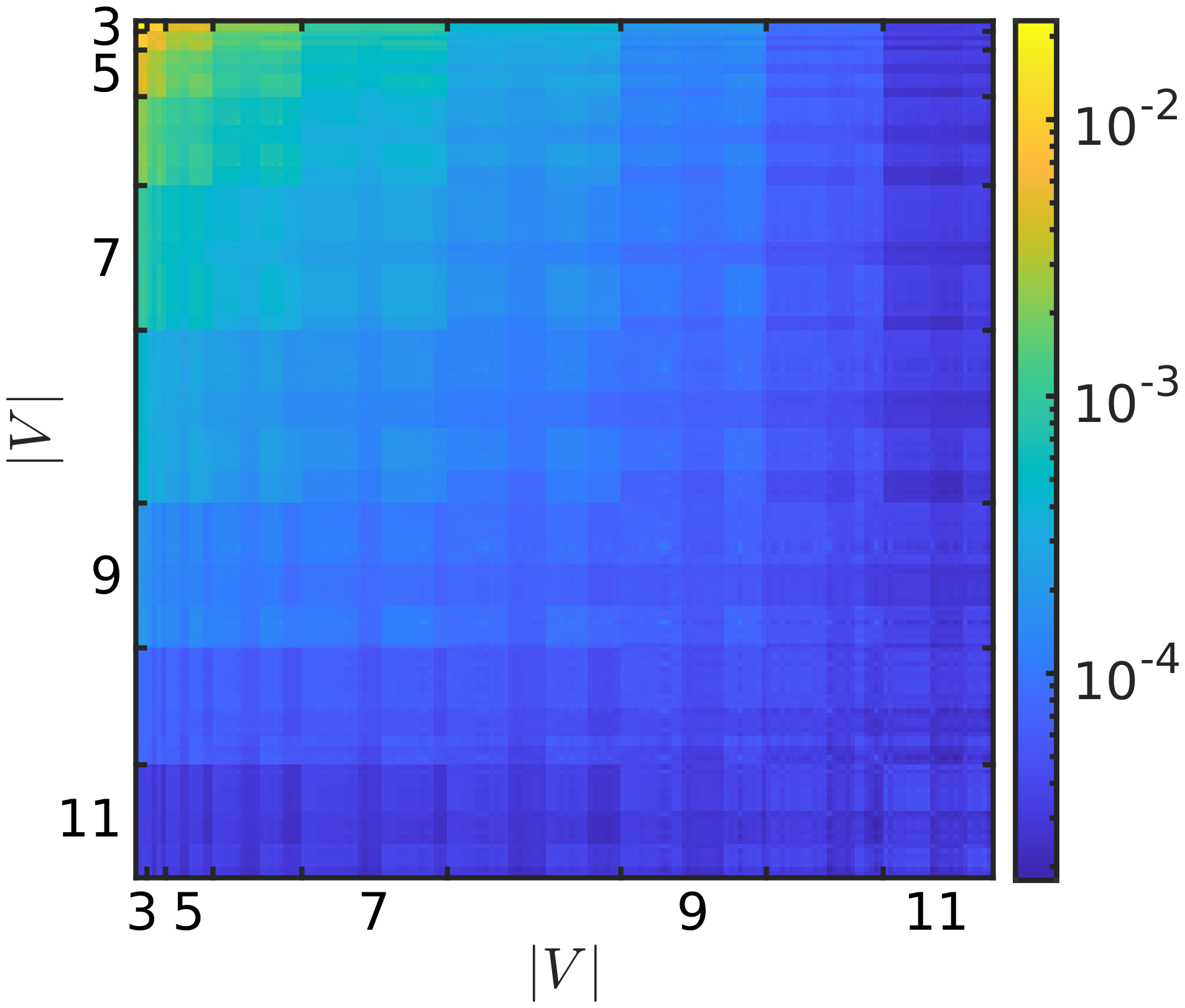}
    \end{minipage}
    \caption{(a) Average errors and root-mean-square errors for the test set for different values of $\sigma$. The solid lines depict the median and the semi-transparent areas comprise the 30th to the 70th percentile of the respective errors. (b) Entries of the Gram matrix corresponding to the whole data set for $\sigma=3$ with a logarithmically scaled color map. The samples are sorted by number of atoms, type of compound (oxygen/sulfur), and number of contained heteroatoms.}
    \label{fig:acyclic_results}
\end{figure}

The best results in terms of the average and the root-mean-square error are obtained for kernel parameters $\sigma$ between $2.5$ and $2.8$, see Table~\ref{tab:acyclic} for details. According to the data set's website, the best results so far for the boiling point prediction on $90\%$ of the set as training data and $10\%$ as test data are achieved by applying so-called \emph{treelet kernels} which exploit all possible graph/tree patterns up to a given size \cite{GAUZERE2012}. In this case, the average error is listed as $4.87$ and the root-mean-square error as $6.75$. Both values are comparable with our results.

\begin{table}
\centering
\caption{Mean and median of the average error and root-mean-square error over all repetitions for values of $\sigma$ between $2.5$ and $2.8$. }
\label{tab:acyclic}
\begin{tabular}{c|cc|cc}
 \multirow{2}{*}{$\sigma$} & \multicolumn{2}{c|}{average error} & \multicolumn{2}{c}{root-mean-square error} \\\cline{2-5}
  & ~~~~mean~~~~ & ~~~~median~~~~ & ~~~~mean~~~~ & ~~~~median~~~~ \\\hline
  $0.25$ & $4.90$ & $4.76$ & $6.85$ & $6.57$ \\
  $0.26$ & $4.91$ & $4.77$ & $6.87$ & $6.53$ \\
  $0.27$ & $4.92$ & $4.77$ & $6.88$ & $6.47$ \\
  $0.28$ & $4.94$ & $4.80$ & $6.91$ & $6.44$ \\
\end{tabular}
\end{table}

As shown in Figure~\ref{fig:acyclic_results}~(b), the entries of the Gram matrix tend to decrease for larger molecules. This effect can be explained by the expansion of the adjacency matrices and similarities of the molecules with small numbers of atoms. For instance, the first two compounds in the (ordered) data set are dimethyl ether (C$_2$H$_6$O) and dimethyl sulfide (C$_2$H$_6$S). Due to the expansion from $\abs{V}=3$ to $\abs{V}=11$, the majority of the permutations in~\eqref{eq:graph kernel by hyperpermanent} do not affect the adjacency matrix of $G^\prime$, cf.~Appendix~\ref{app:isolated}. The block structure of the matrix arises from the ordering of the data set, i.e., each group of molecules with the same number of atoms is divided into subgroups of compounds containing one oxygen atom, two oxygen atoms, one sulfur atom, and two sulfur atoms.

\section{Conclusion}
\label{sec:Conclusion}

We derived symmetric and antisymmetric kernels that can be used in kernel-based learning algorithms such as kernel PCA, kernel CCA, or support vector machines, but also to approximate symmetric or antisymmetric eigenfunctions of transfer operators or differential operators (e.g., the Koopman generator or Schrödinger operator). Potential applications range from point cloud analysis and graph classification to quantum physics and chemistry. Furthermore, we analyzed the induced reproducing kernel Hilbert spaces and resulting feature space dimensions. The effectiveness of the proposed kernels was demonstrated using guiding examples and simple benchmark problems.

The next step is now to apply kernel-based methods to more complex quantum systems. Such problems might require kernels tailored to the system at hand. By exploiting additional properties (sparsity, low-rank structure, weak coupling between subsystems), it could be possible to improve the performance of kernel-based methods. Furthermore, the kernel flow approach proposed in \cite{Owhadi19} could be extended to operator estimation problems. This would allow us to also learn the kernel from data.

Another topic for future research would be to consider other types of symmetries and to develop kernels that explicitly take these properties into account. While the antisymmetric kernel can be evaluated efficiently using matrix factorizations, this is not possible for the symmetric kernel, which requires the evaluation of a matrix permanent. Utilizing efficient approximation schemes could speed up the generation of the required Gram matrices significantly. Alternatively, the product or quotient formulation of symmetric kernels could be exploited to facilitate the application of the proposed methods to higher-dimensional problems.

\section*{Acknowledgments}

We would like to thank Jan Hermann for helpful discussions about quantum chemistry and the reviewers for their helpful comments and suggestions.

\section*{Funding}

P.~Gelß and F.~No\'e have been partially funded by Deutsche Forschungsgemeinschaft (DFG) through grant CRC 1114 \emph{``Scaling Cascades in Complex Systems''} (project ID: 235221301, projects A04 and B06). F.~No\'e also acknowledges funding from BMBF through the Berlin Institute for the Foundations of Learning and Data (BIFOLD), the European Commission (ERC CoG 772230), and the Berlin Mathematics center MATH+ (AA2-8).

\section*{Data availability}

The data and code that support the findings of this study are openly available at \url{https://github.com/sklus/d3s/}.

\bibliographystyle{unsrturl}
\bibliography{AntisymmetricKernels}

\appendix

\section{Hyperpermanents}
\label{app:hyperpermanents}

Given a tensor $T \in \R^{d^{\times 4}} = \R^{d \times d \times d \times d}$, the hyperpermanent of $T$ is given by 
\begin{equation*}
  \hper(T) = \sum_{\pi \in S_d} \prod_{i=1}^d \prod_{j=1}^d t_{\pi(i),\pi(j), i, j} = \sum_{\pi \in S_d} \prod_{i=1}^d \prod_{j=1}^d t_{i,j,\pi(i),\pi(j)}.
\end{equation*}
In Example~\ref{ex:graphs}, we considered tensor entries of the form
\begin{equation}\label{eq:Gaussian kernel}
    t_{i,j,k,l} = \exp(-\frac{\big(a_{i,j} - a^\prime_{k,l}\big)^2}{2 \sigma^2})
\end{equation}
for adjacency matrices $A, A^\prime \in \R^{d \times d}$ in order to construct the symmetrized Gaussian kernel for graphs. Another example for defining the entries of $T$, as described in Section~\ref{sec:acyclic}, is
\begin{equation}\label{eq:Laplacian kernel}
  t_{i,j,k,l} = \exp(-\frac{\abs{a_{i,j} - a^\prime_{k,l}}}{\sigma}),
\end{equation}
which results in the symmetrized Laplacian kernel for graphs. For both choices, we set $t_{i,i,k,k} = \exp(-1/2 \sigma^2)$ and $t_{i,i,k,k} = \exp(-1/\sigma)$, respectively, if $a_{i,i} \neq a^\prime_{k,k}$ and $t_{i,i,k,k} = \exp(0)$, otherwise. In what follows, we will consider  different techniques for computing the hyperpermanent of $T$.

\subsection{Laplace expansion for the computation of hyperpermanents}
\label{app:Laplace}

Define $T^{(\mu)}  \in \R^{d^{\times 4}}$ by
\begin{equation*}
t^{(\mu)}_{i,j,k,l} =
\begin{cases}
t_{i,j,k,l} \cdot  t_{i,1,k,\mu} \cdot t_{1,j,\mu,l}, & \textrm{if } i=j \textrm{ and } k=l,\\
t_{i,j,k,l}, & \textrm{otherwise.}
\end{cases}
\end{equation*}
Let $\widehat{T}^{(\mu)} \in \R^{(d-1)^{\times 4}}$ denote the tensor that results from $T^{(\mu)}$ by removing all entries $t^{(\mu)}_{i,j,k,l}$ with $i=1$, $ j=1$, $k=\mu$, or $l = \mu$. The hyperpermanent can then be written as
\begin{align*}
    \hper(T) &= \sum_{\pi \in S_d} \ts \prod_{i=1}^d \ts \prod_{j=1}^d t_{i, j, \pi(i), \pi(j)}\\
    &= \sum_{\pi \in S_d} t_{1, 1, \pi(1), \pi(1)} \prod_{i=2}^d t_{i, 1, \pi(i), \pi(1)} \prod_{j=2}^d t_{1,j,\pi(1),\pi(j)} \prod_{i=2}^d \ts \prod_{j=2}^d t_{i, j, \pi(i), \pi(j)}  \\
    &= \sum_{\mu=1}^d t_{1, 1, \mu, \mu} \sum_{\substack{\pi \in S_d\\ \pi(1) = \mu} } \prod_{i=2}^d \ts \prod_{j=2}^d t_{i, j ,\pi(i), \pi(j)} \left ( t_{i, 1, \pi(i), \mu} \ts t_{1,j,\mu,\pi(j)} \right )^{\delta_{ij}}\\
    &= \sum_{\mu=1}^d t_{1, 1, \mu, \mu} \hper \left(\widehat{T}^{(\mu)} \right),
\end{align*}
where $\delta_{ij}$ denotes the Kronecker delta. Note that $i=j$ implies $\pi(i) = \pi(j)$.

\subsection{Hyperpermanents of pairwise symmetric tensors}

Suppose $t_{i,j,k,l} = t_{j,i,k,l}$ and $t_{i,j,k,l} = t_{i,j,l,k}$ for all $i,j,k,l \in \{1, \dots, d\} $. For instance, this is the case for tensors $T$ containing elementwise evaluations of Gaussian and Laplacian kernels as given in \eqref{eq:Gaussian kernel} and \eqref{eq:Laplacian kernel}, respectively. The hyperpermanent of $T$ can then be written as
\begin{equation} \label{eq:hper symmetric}
\begin{split}
    \hper(T) &= \sum_{\pi \in S_d} \ts \prod_{i=1}^d \ts \prod_{j=1}^d t_{i, j, \pi(i), \pi(j)}\\
    &= \sum_{\pi \in S_d} \ts \prod_{i=1}^d \left( t_{i, i, \pi(i), \pi(i)} \ts \prod_{\substack{j=1 \\ j \neq i}}^d t_{i, j, \pi(i), \pi(j)} \right)\\
    &= \sum_{\pi \in S_d} \ts \left( \prod_{i=1}^d t_{i, i, \pi(i), \pi(i)} \right) \left( \prod_{i=1}^d \ts \prod_{\substack{j=1 \\ j \neq i}}^d t_{i, j, \pi(i), \pi(j)} \right)\\
    &= \sum_{\pi \in S_d} \ts \left( \prod_{i=1}^d t_{i, i, \pi(i), \pi(i)} \right) \left(\prod_{i=1}^d \ts \prod_{j = i+1}^d t_{i, j, \pi(i), \pi(j)}^2 \right)\\
    &= \sum_{\pi \in S_d} \frac{ \prod_{i=1}^d \ts \prod_{j = i}^d t_{i, j, \pi(i), \pi(j)} ^2}{\prod_{i=1}^d t_{i, i, \pi(i), \pi(i)}}.
\end{split}
\end{equation}
The advantage of the above formula is that we can reduce the computational costs for the hyperpermanent. Additionally, we do not have to compute all elements of $T$, which also reduces the computational costs.

\subsection{Hyperpermanents for graphs with isolated nodes}
\label{app:isolated}

In order to compare graphs of different sizes, we include artificial isolated nodes in Section~\ref{sec:acyclic}. That is, the adjacency matrix of a given graph is expanded by adding zero entries. In this case, all permutations among the isolated nodes do not change the result of the product $\prod_{i=1}^d \ts \prod_{j=1}^d t_{i, j, \pi(i), \pi(j)}$. Assume that the dimension of the adjacency matrix $A$ used in~\eqref{eq:Gaussian kernel} and~\eqref{eq:Laplacian kernel} is initially $d^\prime < d$ before the expansion to $R^{d \times d}$. Then, it holds that
\begin{equation*}
 t_{i,j,\pi(i), \pi(j)} = t_{d^\prime+1, d^\prime +1, \pi(i), \pi(j)},
\end{equation*}
if $i>d^\prime$ or $j>d^\prime$. Thus, given two permutations $\pi_1$ and $\pi_2$ with
\begin{equation*}
 \pi_1(i) = \pi_2(i) \text{ for } 1 \leq i \leq d^\prime, 
\end{equation*}
it follows that
\begin{equation*}
    \frac{ \prod_{i=1}^d \ts \prod_{j = i}^d t_{i, j, \pi_1(i), \pi_1(j)} ^2}{\prod_{i=1}^d t_{i, i, \pi_1(i), \pi_1(i)}} = \frac{ \prod_{i=1}^d \ts \prod_{j = i}^d t_{i, j, \pi_2(i), \pi_2(j)} ^2}{\prod_{i=1}^d t_{i, i, \pi_2(i), \pi_2(i)}}.
\end{equation*}
This means that for each permutation $\pi \in S_d$, any of the $(d-d^\prime)!$ permutations of the set ${\{\pi(d^\prime+1), \dots, \pi(d)\}}$ does not change the value of the quotient in~\eqref{eq:hper symmetric}. This fact can be exploited using the formula
\begin{equation*}
    \hper(T) = \sum_{\mathclap{\substack{\pi \in S_d\\ \pi(d^\prime+1) < \dots < \pi(d)}}} (d-d^\prime)! \, \frac{ \prod_{i=1}^d \ts \prod_{j = i}^d t_{i, j, \pi(i), \pi(j)} ^2}{\prod_{i=1}^d t_{i, i, \pi(i), \pi(i)}},
\end{equation*}
which enables us to decrease the number of considered permutations significantly if $d^\prime$ is much smaller than $d$.

\section{Quotient representation of symmetric Gaussian kernels}\label{app: quotient kernels}

Under the same assumptions as given in Example~\ref{ex: quotient Gauss kernel}, we use Lemma~\ref{lem:Gaussian Slater kernel} in order to write $k_s(x,x^\prime)$ as
\begin{equation*}
 k_s(x,x^\prime) = \frac{k^{(1)}_a(x,x^\prime)}{k^{(2)}_a(x,x^\prime)} = \frac{e^{-\frac{(x_1 - x_1^\prime)^2}{2 \sigma_1^2}} e^{-\frac{(x_2 - x_2^\prime)^2}{2 \sigma_1^2}} - e^{-\frac{(x_1 - x_2^\prime)^2}{2 \sigma_1^2}} e^{-\frac{(x_2 - x_1^\prime)^2}{2 \sigma_1^2}}}{e^{-\frac{(x_1 - x_1^\prime)^2}{2 \sigma_2^2}} e^{-\frac{(x_2 - x_2^\prime)^2}{2 \sigma_2^2}} - e^{-\frac{(x_1 - x_2^\prime)^2}{2 \sigma_2^2}} e^{-\frac{(x_2 - x_1^\prime)^2}{2 \sigma_2^2}}},
\end{equation*}
where $x_1 \neq x_2$ and $x_1 ^\prime \neq x_2^\prime$. From l'H{\^o}pital's rule follows that
\begin{align*}
   \lim_{x_2 \rightarrow x_1} \frac{k^{(1)}_a(x,x^\prime)}{k^{(2)}_a(x,x^\prime)} &= \lim_{x_2 \rightarrow x_1} \frac{\frac{\partial}{\partial x_2} k^{(1)}_a(x,x^\prime)}{\frac{\partial}{\partial x_2} k^{(2)}_a(x,x^\prime)} \\
   &= \lim_{x_2 \rightarrow x_1} \frac{-\frac{x_2 - x_2^\prime}{\sigma_1^2} \, e^{-\frac{(x_1 - x_1^\prime)^2}{2 \sigma_1^2}} \, e^{-\frac{(x_2 - x_2^\prime)^2}{2 \sigma_1^2}} + \frac{x_2 - x_1^\prime}{\sigma_1^2} \, e^{-\frac{(x_1 - x_2^\prime)^2}{2 \sigma_1^2}} \, e^{-\frac{(x_2 - x_1^\prime)^2}{2 \sigma_1^2}}}{-\frac{x_2 - x_2^\prime}{\sigma_2^2} \, e^{-\frac{(x_1 - x_1^\prime)^2}{2 \sigma_2^2}} \, e^{-\frac{(x_2 - x_2^\prime)^2}{2 \sigma_2^2}} + \frac{x_2 - x_1^\prime}{\sigma_2^2} \, e^{-\frac{(x_1 - x_2^\prime)^2}{2 \sigma_2^2}} \, e^{-\frac{(x_2 - x_1^\prime)^2}{2 \sigma_2^2}}}\\
   &= \frac{ \left( -\frac{x_1 - x_2^\prime}{\sigma_1^2} + \frac{x_1 - x_1^\prime}{\sigma_1^2} \right) \, e^{-\frac{(x_1 - x_1^\prime)^2}{2 \sigma_1^2}} \, e^{-\frac{(x_1 - x_2^\prime)^2}{2 \sigma_1^2}}}{\left( -\frac{x_1 - x_2^\prime}{\sigma_2^2} + \frac{x_1 - x_1^\prime}{\sigma_2^2} \right) \, e^{-\frac{(x_1 - x_1^\prime)^2}{2 \sigma_2^2}} \, e^{-\frac{(x_1 - x_2^\prime)^2}{2 \sigma_2^2}}}\\
   &= \frac{\sigma_2^2}{\sigma_1^2} \, \frac{ e^{-\frac{\lVert (x_1, x_1) - (x_1^\prime, x_2^\prime) \rVert^2}{2 \sigma_1^2}}}{e^{-\frac{\lVert (x_1, x_1) - (x_1^\prime, x_2^\prime) \rVert^2}{2 \sigma_2^2}}}.
\end{align*}

\end{document}